\documentclass[a4paper,USenglish,cleveref,numberwithinsect,thm-restate]{lipics-v2021}

\usepackage{microtype} 

\usepackage[obeyFinal]{todonotes}
\usepackage{csquotes}
\usepackage{placeins}
\usepackage{cite}

\usepackage{caption}
\usepackage{subcaption}

\graphicspath{{./figures/}}

\DeclareMathOperator{\polylog}{polylog}
\newcommand{\OPT}{\ensuremath{\mathrm{OPT}}} %
\newcommand{\OPTcr}{\ensuremath{\mathrm{OPT}_{\scriptscriptstyle\mathrm{cr}}}} %
\newcommand{\Topt}{\ensuremath{T_{\scriptscriptstyle\mathrm{OPT}}}} %
\newcommand{\Toptd}{\Topt^\star}
\newcommand{\Talg}{\ensuremath{T_{\scriptscriptstyle\mathrm{ALG}}}} %
\newcommand{\Tcr}{\ensuremath{T_{\scriptscriptstyle\mathrm{cr}}}} %
\usepackage{dsfont}
\newcommand{\R}{\ensuremath{\mathds{R}}} %

\newcommand\pe{\mathcal{P}}
\DeclareMathOperator{\diam}{diam}

\def\ds#1{\overrightarrow{#1}}
\newcommand{\eps}{\varepsilon}

\def\Td#1{T^{#1}_{\scriptscriptstyle\mathrm{OPT}}} 
\def\bd{\operatorname{BoundDiam}} 
\newcommand{\bigO}{\mathcal{O}}

\theoremstyle{definition}

\title{Long Plane Trees\thanks{A preliminary version appeared as S.~Cabello,
M.~Hoffmann, K.~Klost, W.~Mulzer, and J. Tkadlec. \emph{Long Plane Trees}
in Proc.~38th SoCG, pp. 23:1–-23:17.}}

\titlerunning{Long Plane Trees}

\author{Sergio Cabello}{Faculty of Mathematics and Physics, University of Ljubljana, Slovenia\and
Institute of Mathematics, Physics and Mechanics, Slovenia}{sergio.cabello@fmf.uni-lj.si }{https://orcid.org/0000-0002-3183-4126}{Supported in part by the Slovenian Research Agency (P1-0297, J1-9109, J1-8130, J1-8155, J1-1693, J1-2452, N1-0218, N1-0285). Supported in part by the European Union (ERC, KARST, project number 101071836). Views and opinions expressed are however those of the authors only and do not necessarily reflect those of the European Union or the European Research Council. Neither the European Union nor the granting authority can be held responsible for them.}

\author{Michael Hoffmann}{Department of Computer Science, ETH Z\"urich, Switzerland}{hoffmann@inf.ethz.ch}{https://orcid.org/0000-0001-5307-7106}{Supported by the Swiss National Science Foundation within the collaborative DACH project \emph{Arrangements and Drawings} as SNSF Project 200021E-171681.}
	
\author{Katharina Klost}{Institut f\"ur Informatik, Freie Universit\"at Berlin, Germany}{kathklost@inf.fu-berlin.de}{}{}

\author{Wolfgang Mulzer}{Institut f\"ur Informatik, Freie Universit\"at Berlin, Germany}{mulzer@inf.fu-berlin.de}{https://orcid.org/0000-0002-1948-5840}{Supported in part by ERC StG 757609.}
  
\author{Josef Tkadlec}{Computer Science Institute, Charles University, Czech Republic}{josef.tkadlec@iuuk.mff.cuni.cz}{https://orcid.org/0000-0002-1097-9684}{Supported by Charles Univ.~projects UNCE 24/SCI/008 and PRIMUS 24/SCI/012.}

\authorrunning{S. Cabello, M. Hoffmann, K. Klost, W. Mulzer, and J. Tkadlec}

\Copyright{Sergio Cabello, Michael Hoffmann, Katharina Klost, 
Wolfgang Mulzer and Josef Tkadlec} 



\hideLIPIcs

\nolinenumbers

\begin{document}
\maketitle

\begin{abstract}
In the \emph{longest plane spanning tree} problem, we 
are given a finite planar point set $\pe$, and our task 
is to find a plane (i.e., noncrossing) spanning tree 
for $\pe$ with maximum total Euclidean edge 
length. Despite more than two decades of 
research, it remains open whether this problem is 
NP-hard. Thus, previous efforts have focused on 
polynomial-time algorithms that produce plane trees 
whose total edge length approximates $\OPT$, the maximum
possible length. 
The approximate trees in these
algorithms all have 
small unweighted diameter, typically three or four. 
It is natural to ask whether this is a common 
feature of longest plane spanning trees, or 
an artifact of the specific approximation algorithms.

We provide three results to elucidate the
interplay between the approximation guarantee and the 
unweighted diameter of the approximate trees.
First, we describe a polynomial-time algorithm to construct a plane tree 
with diameter at most four 
and total edge length at least $0.546 \cdot \OPT$.
This constitutes a substantial improvement over the state
of the art.
Second, we show that a longest plane tree among those with 
diameter at most three 
can be found in polynomial time.
Third, for any candidate diameter $d \geq 3$, we 
provide upper bounds on the approximation factor 
that can be achieved by a longest plane tree 
with diameter at most $d$ (compared to a longest plane tree
without constraints). 
\end{abstract}


\section{Introduction}\label{sec:intro}
\emph{Geometric network design} is a common and well-studied
task in computational geometry and combinatorial 
optimization~\cite{MM17,Har-Peled11,Eppstein00,Mitchell04}. 
In this family of problems, we are given a set $\pe$ of points
in general position, 
and our task is to connect $\pe$ 
into a (geometric) graph that has certain favorable properties.
Not surprisingly, this general question has captivated
the attention of researchers for a long time,
and we can find countless variants, depending on which 
restrictions we put on the graph that connects $\pe$ and which
criteria of this graph we would like to optimize.
Typical graph classes of interest include 
matchings, paths, cycles, trees, or general \emph{plane} 
(\emph{noncrossing}) graphs, i.e.,
graphs, whose straight-line embedding on $\pe$ does not contain any
edge crossings.
Typical quality criteria include the 
total edge length~\cite{MulzerR08,Arora98,Mitchell99,deBergChVKOv08},
the maximum length (bottleneck) edge~\cite{EfratIK01,Biniaz20},
the maximum degree~\cite{Chan04,AroraC04,PapadimitriouV84,FranckeH09}, 
the dilation~\cite{mulzer04minimum,NarasimhanSm07,Eppstein00}, 
or the stabbing number~\cite{MulzerOb20,Welzl92}
of the graph.

Many famous problems from computational geometry fall into this 
general setting.
For example, if our goal is to minimize the total edge length, while
restricting ourselves to paths, trees, or triangulations, respectively,
we are faced with the venerable problems of finding an optimum TSP 
tour~\cite{Har-Peled11}, a Euclidean minimum spanning 
tree~\cite{deBergChVKOv08},
or a minimum weight triangulation~\cite{MulzerR08} for $\pe$.  
These three examples also illustrate the wide variety of complexity aspects
that we may encounter in geometric design problems: the Euclidean TSP
problem is known to be NP-hard~\cite{Papadimitriou77}, 
but it admits a PTAS~\cite{Arora98,Mitchell99}. On the other hand,
it is possible to find a Euclidean minimum spanning tree for $\pe$
in polynomial time~\cite{deBergChVKOv08} (even though, curiously, the 
associated decision problem is not known to be solvable by a 
polynomial-time Turing machine, see, e.g., \cite{bloemer_radicals_1991}).
The minimum weight triangulation
problem is also known to be NP-hard~\cite{MulzerR08}, 
but the existence of a PTAS
is still open; however, a QPTAS is known~\cite{RemyS09}.

In this work, we are interested in the interaction of two specific requirements 
for a geometric design problem, namely the two objectives of obtaining
a plane graph and of optimizing the total edge length. For the case that we
want to \emph{minimize} the total edge length of the resulting graph,
these two goals are often in perfect harmony: the shortest Euclidean
TSP tour and the shortest Euclidean minimum spanning tree are 
automatically plane, as can be seen by a simple application 
of the triangle inequality.
In contrast, if our goal is to \emph{maximize} the total
edge length, while obtaining a plane graph, much less is known.

This family of problems was studied by Alon, Rajagopalan, and
Suri~\cite{AlonRS95}, who considered the problems of computing a
longest plane matching, a longest plane
Hamiltonian path, and a longest plane spanning tree for a planar 
point set $\pe$
in general position. They conjectured that
these three problems are all NP-hard, but as far as we know, 
this is still open.
The situation is similar for the problem of finding a 
\emph{maximum weight triangulation} for $\pe$: here, we
have neither an NP-hardness proof, nor a polynomial time 
algorithm~\cite{ChinQW04}.
If we omit the planarity condition, then the problem of finding a
longest Hamiltonian path (the \emph{geometric maximum TSP problem})
is known to be NP-hard in dimension three and above, while the two-dimensional
case remains open~\cite{BarvinokFJTWW03}. 
On the other hand, we can find a longest (typically  
not plane) tree on $\pe$ in polynomial time, using
classic greedy algorithms~\cite{CLRS}. 

\subparagraph*{Longest plane spanning trees.}
We focus on the specific problem of finding a longest plane 
(i.e., noncrossing) tree for a given set $\pe$ of $n\ge 3$ points in the plane
in general position (that is, no three points in $\pe$ are collinear).
Such a tree is necessarily spanning.
The general position assumption was also used in previous 
work on this problem~\cite{AlonRS95,DBLP:journals/dcg/DumitrescuT10}.
Without it, one should specify whether overlapping edges are allowed,
an additional complication that we would like to avoid.

If $\pe$ is in convex position, then the longest plane 
tree for $\pe$ can be found in polynomial time on a real RAM, 
by adapting standard dynamic programming methods for plane
structures on convex point sets~\cite{Gilbert79,Klincsek80}.
On the other hand, for an arbitrary point set $\pe$, the problem
is conjectured---but not known---to be NP-hard~\cite{AlonRS95}.
Hence, past research has focused on designing polynomial-time 
approximation algorithms.
Typically, these algorithms construct 
several ``simple'' spanning trees for $\pe$ of small (unweighted) 
diameter, and one then argues 
that at least one such tree is sufficiently long.
In a seminal work, Alon et
al.~\cite{AlonRS95}  showed that a longest
star (a plane tree with diameter two) on $\pe$ yields a 
$0.5$-approximation for the
longest (not necessarily plane) spanning tree of $\pe$.
They further argued that this bound is
essentially tight for point sets that consist of two large 
clusters far away from each other. 
Dumitrescu and T\' oth~\cite{DBLP:journals/dcg/DumitrescuT10} refined
this algorithm by adding two additional families of candidate trees,
now with diameter four. They showed that at least one member of this
extended range of candidates constitutes a
$0.502$-approximation, which was further improved to $0.503$ by Biniaz et
al.~\cite{biniaz_maximum_2019}.  In all these results, the 
approximation factor is analyzed by
comparing to the length of a longest (typically not plane)
spanning tree.  Such a tree may be longer by a factor of up to $\pi/2>1.5$ 
than a maximum-length plane tree~\cite{AlonRS95}, as witnessed
by, e.g., a large set of points spaced uniformly on a unit circle.
While the ratio between the lengths of the longest plane tree and 
the longest (possibly crossing) tree is an interesting quantity in itself, 
the original objective is to construct a longest plane tree, and
thus, it is better to compare the length of the constructed plane trees
against the true optimum, that is, against the length of a longest plane tree.
Considering certain trees of diameter at most five, a superset 
of the authors of this paper managed to compare 
against the longest plane tree and pushed the approximation factor to
$0.512$~\cite{cabello2020better}. This was subsequently improved even further
to $0.519$ by Biniaz~\cite{biniaz2020improved}.

\subparagraph*{Our results.}
We provide a deeper study of the interplay between the approximation factor 
and the diameter of the candidate trees.
First, we give 
a polynomial-time 
algorithm to find a tree  of diameter at most four
that guarantees an approximation factor of roughly $0.546$, 
a substantial improvement over the previous bounds.

\begin{restatable}{theorem}{thmapproximation}\label{thm:approximation}
  For any finite point set $\pe$ in general position (the set $\pe$
  contains no three collinear points),
  we can compute in polynomial time a plane tree of
  Euclidean length at least $f \cdot \OPT$, where $\OPT$ denotes
  the length of a longest plane tree on $\pe$ and $f> 0.5467$ is the
  fourth smallest real root of the polynomial
  \[P(x)=-80 + 128 x + 504 x^2 - 768 x^3 - 845 x^4 + 1096 x^5 + 256 x^6.
  \]
\end{restatable}

The algorithm of \cref{thm:approximation} 
constructs six plane trees: 
four stars and two general trees of 
diameter at most four. We show 
that one of these 
trees is always sufficiently long.  
The algorithm is very simple, but its analysis
uses several geometric insights. The polynomial $P(x)$ comes from
optimizing several constants that appear in the algorithm.

Second, we
characterize longest plane trees for convex point sets.
A \emph{caterpillar} is a tree $T$ that contains a path $S$, the
\emph{spine}, so that every vertex of $T\setminus S$ is adjacent to a
vertex in $S$.
A tree $T$ that is straight-line embedded on a convex point set $\pe$
is a \emph{zigzagging caterpillar} if its edges split the convex hull of $\pe$
into faces that are all triangles. Our next two theorems show that
the longest plane trees for convex sets are given exactly by caterpillar
graphs.

\begin{restatable}{theorem}{thmconvex}\label{thm:convex}
  Let $\pe$ be a convex finite point set in the plane. Then, 
  every longest plane tree on $\pe$ is a zigzagging
  caterpillar.
\end{restatable}

\begin{restatable}{theorem}{thmcaterpillars}\label{thm:caterpillars}
  For any caterpillar $C$, there exists a convex point set $\pe$ such
  that the unique longest tree for $\pe$ is isomorphic to $C$.
\end{restatable}

In particular, \cref{thm:caterpillars} implies that 
the diameter of a (unique) longest plane tree can be
arbitrarily large.  As a consequence, we obtain an upper bound on the
approximation factor $\bd(d)$ that can be achieved by a plane tree of
diameter at most $d$.

\begin{restatable}{theorem}{thmdiameterbound}\label{thm:diameter-bound}
  For every $d\ge 2$, there exists a convex point set $\pe$ so that
  every plane tree of diameter at most $d$ on $\pe$ is at most
  \[
    \bd(d)\le 1- \frac{6}{(d+1)(d+2)(2d+3)} = 1-\Theta(1/d^3)
  \]
  times as long as the length $|\Topt|$ of a longest (general) plane
  tree on $\pe$.
\end{restatable}

We have better bounds for small values of $d$. For $d=2$, it is easy to
see that $\bd(2)\le 1/2$: put two groups of roughly half of the points
sufficiently far from each other \cite{AlonRS95}. For $d=3$, we can show
$\bd(3)\le 5/6$.

\begin{restatable}{theorem}{thmdiameterthreebound}\label{thm:diameter-3-bound}
  For every $\eps>0$, there exists a convex point set $\pe$ such that every
  longest plane tree of diameter  three on $\pe$ is at most 
  $(5/6)+\eps$ times as long as every longest (general) plane tree on $\pe$.
\end{restatable}

Third, we give polynomial-time algorithms for finding a longest
plane tree 
among those of diameter at most three and among a special class of
trees of diameter at most four.
Note that in contrast to diameter two, the number of spanning trees of
diameter at most three is exponential in the number of points.

\begin{restatable}{theorem}{thmdp}\label{thm:dp}
  For any set $\pe$ of $n$ points in general position, a longest plane 
  tree of diameter at
  most three on $\pe$ can be computed in $\bigO(n^4)$ time.
\end{restatable}

\begin{restatable}{theorem}{thmtristar}\label{thm:tristar}
  For any set $\pe$ of points in general position and for any 
  three specified points on the boundary of the
  convex hull of $\pe$, we can compute in polynomial time 
  the longest plane tree such that each edge
  is incident to at least one of the three specified points.
\end{restatable}

The algorithms are based on dynamic programming. Even though the
length $\OPT_3$ of a longest plane tree of diameter at most three
can be computed in polynomial time, we do not know the corresponding
approximation factor $\bd(3)$ with respect to $\OPT$. 
The best bounds we are aware of are
$1/2\le \bd(3)\le 5/6$. The lower bound follows from~\cite{AlonRS95},
the upper bound is from~\cref{thm:diameter-3-bound}. We conjecture
that $\OPT_3$ actually gives a better approximation factor than the
tree constructed in~\cref{thm:approximation}---but we are unable to
prove this.

Finally, a natural way to design an algorithm for the longest plane
spanning tree problem is the following local search
heuristic~\cite{WilliamsonSh11}:
start with an arbitrary plane tree $T$, and while it is possible, apply the 
following \emph{local improvement rule}: 
if there are two edges $e$,~$f$ on $\pe$
such that $(T \setminus \{e\}) \cup \{f\}$ is a plane spanning tree for
$\pe$ that is longer than $T$, replace $e$ by $f$.
Once no further local improvements are possible, output the current tree $T$.
We show that for some point sets, 
this algorithm fails to compute the optimum answer as it may ``get stuck'' 
in a local optimum (see~\cref{thm:stuck} in \cref{sec:stuck}).
This holds regardless of how the edges that are swapped are chosen.
This suggests that a natural local search approach does not yield an
optimal algorithm for the problem.


\subparagraph*{Preliminaries and Notation.}\label{sec:preliminaries}
Let $\pe\subset\R^2$ be a set of $n$ points in the plane, 
so that no three points in $\pe$ are collinear.
For any spanning tree $T$ on $\pe$, we denote by $|T|$ the 
total Euclidean edge length of $T$. 
Let $\OPT$ be the maximum Euclidean length of a plane (i.e., noncrossing) 
spanning tree on $\pe$, and $\OPTcr$ the maximum length of a 
(possibly crossing) spanning tree on $\pe$.

As in previous algorithms~\cite{AlonRS95,
  DBLP:journals/dcg/DumitrescuT10,biniaz_maximum_2019,cabello2020better,biniaz2020improved}, 
  we make extensive use of stars. 
For any point $p \in \pe$, the \emph{star}
$S_p$ rooted at $p$ is the tree that connects $p$ 
to all other points of $\pe$.

We also need a notion of ``flat'' point sets.
A point set $\pe$ is \emph{flat} if $\diam(\pe)\ge 1$ and
all $y$-coordinates in $\pe$ are essentially negligible, that is, 
their absolute values are bounded by an infinitesimal $\eps>0$.
In a flat point set, we can approximate the length of an edge by 
subtracting the $x$-coordinates of its endpoints: the error 
becomes arbitrarily small as $\eps\to 0$.

Lastly, $D(p,r)$ denotes a closed disk with center $p$ and radius $r$,
while $\partial D(p,r)$ is its boundary: a circle of radius $r$ centered at $p$.


\section{A general approximation algorithm}
\label{sec:approximation}
We present a polynomial-time algorithm that yields an 
$ f\doteq 0.5467$-approximation of a longest plane tree
for general point sets and a
$(2/3)$-approximation for flat point sets.

Let $\pe$ be the input point set (in general position). Our
algorithm considers two kinds of plane spanning trees for
$\pe$. First, for every point $a \in \pe$, let $S_a$ be the
star rooted at~$a$, as defined in \cref{sec:intro}.
Second, for any distinct points~$a, b \in \pe$, we define
a tree
$T_{a,b}$ as follows
(see \cref{fig:approx-Tab}):
let $\pe_a$ be the points of $\pe$ that are closer to $a$ than to $b$, and 
let $\pe_b=\pe\setminus \pe_a$.
First, we connect $a$ to every point in $\pe_b$. 
Then, we connect each point of $\pe_a\setminus \{ a \}$ to some point of 
$\pe_b$, according to the  following
rule: 
the rays $\overrightarrow{av}$ for $v\in \pe_b$ together with the 
opposing ray of $\overrightarrow{ab}$ partition the plane into 
convex wedges with common apex $a$.
For each such wedge $W$, let  $b_W \in \pe_b$ the point on a 
bounding ray of $W$ that forms the smaller (convex) angle with 
$\overrightarrow{ab}$.
Within each wedge $W$, we connect all points of 
$W \cap (\pe_a \setminus \{ a \})$ to $b_W$.
The resulting tree $T_{a,b}$  has diameter at most four, and
it is plane, because the interiors of the wedges are pairwise 
disjoint, and we add a star within each wedge.

In general, the trees $T_{a,b}$ and $T_{b,a}$ are different. 
However, for $\pe_a=\{a\}$,
both $T_{a,b}$ and $T_{b,a}$ coincide with the star $S_a$.

\begin{figure}[ht]
\center
\includegraphics[scale=1]{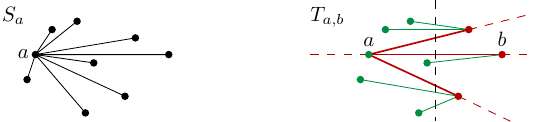}
\caption{A tree $S_a$ and a tree $T_{a,b}$.
}
\label{fig:approx-Tab}
\end{figure}

Now, our algorithm
$\texttt{AlgSimple}(\pe)$
constructs the $n$ stars $S_a$, for $a \in \pe$, and
the $n(n-1)$ trees $T_{a,b}$, for distinct $a, b \in \pe$,
and it computes the total edge length for each of them.
It returns the longest of these $n^2$ trees as the desired
approximation. We denote this resulting tree as $\Talg$.

The algorithm 
$\texttt{AlgSimple}(\pe)$
runs in polynomial time on a real RAM, as there are $n^2$ relevant trees, 
each of which can be constructed in polynomial time.

It remains to analyze the approximation guarantee of \texttt{AlgSimple}.
The main result of this section (\cref{thm:approximation}) states 
that for any point set $\pe$, we have $|\Talg|> 0.5467 \cdot \OPT$.
The proof is rather involved.
As a warm-up for the full proof,
we first show a stronger result for the special case 
of \emph{flat} point sets (cf.~\cref{sec:intro}):
if $\pe$ is flat, we have $|\Talg|\ge (2/3)\cdot \OPTcr$, 
where $\OPTcr$ is the length of a longest (possibly crossing) spanning tree.

\begin{theorem}\label{thm:approximation-flat} 
Suppose $\pe$ is flat. Then $|\Talg| \geq \frac{2}{3} \cdot \OPTcr 
\geq \frac{2}{3} \cdot \OPT$.
\end{theorem}

\begin{proof}
Since $\OPTcr \geq \OPT$, it suffices to prove the first inequality.
Recall that in a flat point set, the $y$-coordinates of the points
are essentially negligible. Thus,
let $a$ be the leftmost point and $b$ be the rightmost point of $\pe$,
and let $m = (a + b)/2$ be the midpoint between $a$ and $b$.
Let $\Tcr$ be 
the tree on $\pe$ where
all points to the right of $m$ are connected to $a$, and
all points to the left of $m$ are connected to 
$b$.\footnote{In particular, $a$ is connected to $b$. If $m$ happens to
belong to $\pe$, it can be connected to either $a$ or $b$.}
Then, $\Tcr$ is a tree that achieves $\OPTcr$,
since
for any edge $pq$ that is not in the $\Tcr$,
	the \emph{fundamental} cycle that $pq$ makes with $\Tcr$ has $pq$ as a
shortest edge.
Now, to analyze the quality of the approximation,
we consider the four trees $S_a$, $T_{a,b}$, $T_{b,a}$, and $S_b$.
Our goal is to show that
\begin{equation}
\label{equ:maxflatbound}
\max \{|S_a|,|T_{a,b}|, |T_{b,a}|, |S_b|\} \geq \frac{2}{3} \cdot \OPTcr,
\end{equation}
which implies the desired result.
For this, we define a probability distribution $\mathcal{D}$
on $\{S_a,T_{a,b}, T_{b,a}, S_b\}$, and we show
that 
\begin{equation}
\label{equ:expflatbound}
\mathbf{E}_{T \sim \mathcal{D}}[|T|] \geq \frac{2}{3} \cdot \OPTcr,
\end{equation}
which immediately implies 
(\ref{equ:maxflatbound}).
For (\ref{equ:expflatbound}), we use the distribution $\mathcal{D}$ with
$\Pr_\mathcal{D}[S_a] = \Pr_\mathcal{D}[S_b] = 1/6$ and
$\Pr_\mathcal{D}[T_{a,b}] = \Pr_\mathcal{D}[T_{b,a}] = 1/3$.
Our strategy is to consider the edges individually and to use 
linearity of expectation to achieve the total bound. 
In particular, for $p \in \pe \setminus \{a \}$ and 
$T \in \{S_a,T_{a,b}, T_{b,a}, S_b, \Tcr\}$, let $\ell_T(p)$
be the length of the first edge on the path from $p$ to
$a$ in $T$. By the definition of $\Tcr$,
we have
\begin{equation}
\label{equ:Tcrell}
  \ell_{\Tcr}(p) = \max\{\|pa\|, \|pb\|\},
\end{equation}
for all $p \in \pe \setminus \{a\}$.
Next, we analyze 
$\mathbf{E}_{T \sim \mathcal{D}}[\ell_T(p)]$, for $p \in \pe \setminus \{a\}$.
Suppose first that $p$ lies to the right of $m$.
Since $S_a$ and $S_b$ are stars, we have $\ell_{S_a}(p) = \|pa\|$ and 
$\ell_{S_b}(p) = \|pb\|$. Furthermore, by our assumption on
$p$, we have $\ell_{T_{a,b}}(p) = \|pa\|$, and
$\ell_{T_{b,a}}(p) \geq \|mp\|$. Now, since $m$ is the
midpoint between $a$ and $b$, we
have $\|am\| = \|mb\|$, and hence
\begin{equation}
\label{equ:papb}
\|pa\| = \|am\| + \|mp\| = \|mb\| + \|mp\| = \|mp\| + \|pb\|  + \|mp\|
= 2\|mp\| + \|pb\|.
\end{equation}
Altogether, this gives
\begin{align*}
\mathbf{E}_{T \sim \mathcal{D}}[\ell_T(p)] &=
\frac{1}{6} \cdot \left(
\ell_{S_a}(p)  + 
2\, \ell_{T_{a,b}}(p) + 
2\, \ell_{T_{b,a}}(p) + 
\ell_{S_b}(p)\right) & \text{(by the definition of $\mathcal{D}$)}  \\
& \geq
\frac{1}{6} \cdot \left(
3\, \|pa\| + 2\, \|mp\| +  \|pb\| \right)
& \text{(by the discussion above)} 
\\
&=
\frac{1}{6} \cdot \left(
3\, \|pa\| + \|pa\| \right) & \text{(by (\ref{equ:papb}))} \\
&= \frac{2}{3}\cdot \|pa\|,
\end{align*}
Similarly, if $p$ lies to the left of $m$, we get
$\mathbf{E}_{T \sim \mathcal{D}}[\ell_T(p)] \geq (2/3) \cdot \|pb\|$, 
Altogether, we use 
(\ref{equ:Tcrell}) to obtain
\begin{equation}
\label{equ:ellbound}
\mathbf{E}_{T \sim \mathcal{D}}[\ell_T(p)] \geq
\frac{2}{3} \cdot \max\{\|pa\|, \|pb\|\} = 
\frac{2}{3} \cdot \ell_{\Tcr}(p),
\end{equation}
for all $p \in \pe \setminus \{ a \}$, 
Now, we can put everything together:
\begin{align*}
\mathbf{E}_{T \sim \mathcal{D}}[|T|] 
&= 
\mathbf{E}_{T \sim \mathcal{D}}\left[\sum_{p \in \pe \setminus \{a \}} 
\ell_T(p)\right]
& \text{(by the definition of $\ell_T(p)$)}
\\
&= 
\sum_{p \in \pe \setminus \{a \} }
\mathbf{E}_{T \sim \mathcal{D}}\left[\ell_T(p)\right]
& \text{(by linearity of expectation)} \\
&\geq
\sum_{p \in \pe \setminus \{a \} }
\frac{2}{3} \cdot 
\ell_{\Tcr}(p) 
= 
\frac{2}{3} \OPTcr. & \text{(by (\ref{equ:ellbound}) and
the definition of $\ell_{Tcr}(p)$)}
\end{align*}
Thus, we have proved 
(\ref{equ:expflatbound}), and the theorem follows.
\end{proof}

In fact, one can find an example where the constant $2/3$ 
is asymptotically tight when comparing to $\OPTcr$:

\begin{lemma}\label{lem:tight}
There is a sequence of point sets $\pe_1,\pe_2, \dots$ with 
$|\pe_n| = n + 1$ and
\[
\lim_{n \to \infty} \frac{\OPT(\pe_n)}
{\OPTcr(\pe_n)} \leq \frac{2}{3}.
\]
\end{lemma}

\begin{proof}
For $n \geq 1$, the set $\pe_n = \{p_0, p_1, \dots, p_{n}\}$ 
is the flat point set where the $x$-coordinate of $p_i$ is
$i$, for $i = 0, \dots, n$, and the $p_i$ are in upper convex
position,
see \cref{fig:thin-cor}.

First, we
argue by induction on $n$ that the star $S_{p_0}$ rooted at $p_0$ is 
a longest plane spanning tree for $\pe_n$, for any $n \geq 1$.
For $n = 1$, this is clear, since $S_{p_0}$ is the only spanning tree for
$\pe_1$.
\begin{figure}[tb]
	\centering
	\includegraphics[page=1]{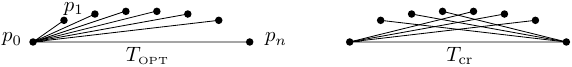}
	\caption{The point set $\pe_n$ consisting of $n + 1$ points with 
	equally spaced $x$-coordinates $0,\dots,n$, with 
	its best plane and crossing spanning trees.}
	\label{fig:thin-cor}
\end{figure}
Now, suppose that $n \geq 2$.
We observe that any longest plane spanning tree
for $\pe_n$ must contain the edge $p_0p_{n}$,
because $p_0p_{n}$ is the longest 
edge on $\pe_n$, and it is not crossed by any 
other edge on $\pe_n$.
Furthermore, we claim that in any longest plane spanning tree for $\pe_n$,
either $p_0$ or $p_{n}$ must be a leaf. Indeed, 
consider a plane spanning tree $T$ on $\pe_n$ that contains
the edge $p_0p_{n}$ and that has the additional 
edges $p_0p_i$ and $p_jp_n$, for $0 < i, j < n$.
Since $T$ is plane, we must have $i < j$, and we can 
strictly increase the length of $T$ by 
replacing $p_0p_i$ with $p_ip_{n}$ (if $i < n/2 $) or 
by replacing $p_jp_{n}$ by $p_0 p_j$ (if $j > n/2$). 
By symmetry, we can thus conclude that there is a longest
plane spanning tree for $\pe_n$ that contains the edge $p_0p_n$ and
in which $p_n$ is a leaf. Since $p_0p_n$ does not cross any
other edge on $\pe_n$, it follows that we can obtain a longest
spanning tree for $\pe_n$ by taking a longest spanning tree for
$\pe_{n-1}$ and by adding the edge $p_0p_n$. By the inductive hypothesis,
the star $S_{p_0}'$ on the point set $\pe_{n-1}$ is a longest plane
spanning tree for $\pe_{n-1}$. Hence, it follows that $S_{p_0}$
is a longest plane spanning tree for $\pe_n$. Thus, we get
\[
  \OPT(\pe_n) = |S_{p_0}| = \sum_{i=1}^{n} i = \frac{n^2}{2} + \frac{n}{2}.
\]
Consider now the tree $\Tcr$ where all points $p_i$ with 
$0 \leq i \leq \lfloor n/2 \rfloor$ are 
connected to $p_n$ and all points $p_i$ with 
$\lfloor n/2 \rfloor + 1 \leq i \leq n$ are connected
to $p_0$ (in particular, $\Tcr$ contains the edge $p_0p_n$). 
This tree is illustrated in \cref{fig:thin-cor} on the right.
Similar as in the proof of \cref{thm:approximation-flat}, the 
tree $\Tcr$ is a longest (crossing) spanning tree for $\pe_n$,
because every edge $p_ip_j$ that does not appear in $\Tcr$ is a shortest
edge 
in the fundamental cycle that $p_ip_j$ makes with $\Tcr$.
Thus, a straightforward summation gives a lower bound on $\OPTcr(\pe_n)$:
\begin{align*}
\OPTcr(\pe_n) &= \sum_{i = 1}^{\lfloor n/2 \rfloor} (n - i) + 
\sum_{i = \lfloor n/2 \rfloor + 1}^n i
& \text{(by the structure of $\Tcr$)}\\
&= 
\sum_{i = 1}^{\lfloor n/2 \rfloor} 
\left(\left\lceil \frac{n}{2} \right\rceil - 1 + i\right) + 
\sum_{i = \lfloor n/2 \rfloor + 1}^n i 
& \text{(index transformation 
$i \mapsto \left\lfloor \frac{n}{2} \right\rfloor + 1 - i$)}
\\
&= 
\left\lfloor \frac{n}{2} \right\rfloor \cdot 
\left(\left\lceil \frac{n}{2} \right\rceil - 1\right) + 
\sum_{i = 1}^{n}  i 
& \text{(rearranging terms)}\\
\\
&= 
\frac{n^2}{2} +
\left\lfloor \frac{n}{2} \right\rfloor \cdot 
\left\lceil \frac{n}{2} \right\rceil
+ \frac{n}{2} - 
\left\lfloor \frac{n}{2} \right\rfloor
& \text{(evaluating the sum, rearranging terms)}\\
\\
&\geq \frac{3n^2}{4},
\end{align*} 
where the last inequality is due to a case distinction:
if $n$ is even, we have
\begin{align*}
\frac{n^2}{2} +
\left\lfloor \frac{n}{2} \right\rfloor \cdot 
\left\lceil \frac{n}{2} \right\rceil
+ \frac{n}{2} - 
\left\lfloor \frac{n}{2} \right\rfloor
&=
\frac{n^2}{2} +
 \frac{n}{2}  \cdot 
 \frac{n}{2} 
+ \frac{n}{2} - 
 \frac{n}{2}\\
&= \frac{3n^2}{4},
\end{align*}
and if $n$ is odd, we have
\begin{align*}
\frac{n^2}{2} +
\left\lfloor \frac{n}{2} \right\rfloor \cdot 
\left\lceil \frac{n}{2} \right\rceil
+ \frac{n}{2} - 
\left\lfloor \frac{n}{2} \right\rfloor
&=
\frac{n^2}{2} +
 \left(\frac{n}{2} - \frac{1}{2}\right)  \cdot 
\left( \frac{n}{2}  + \frac{1}{2} \right)
+ \frac{n}{2} - 
 \left(\frac{n}{2} - \frac{1}{2}\right)\\
 &= 
\frac{n^2}{2} +
\frac{n^2}{4} -
\frac{1}{4} +
\frac{1}{2}\\
&= \frac{3n^2}{4} + \frac{1}{4}.
\end{align*}
Combining the bounds for $\OPT(\pe_n)$ and $\OPTcr(\pe_n)$, 
we get the desired result:
\begin{equation*}
\lim_{n\to\infty} \frac{\OPT(\pe_n)}{\OPTcr(\pe_n)} \leq 
\lim_{n\to\infty}  \frac{n^2/2+n/2}{3n^2/4} = \frac{2}{3}. \qedhere
\end{equation*} 
\end{proof}

Now we use a similar approach to show the main theorem of this section:

\thmapproximation*
\begin{proof} 
We outline the proof strategy, referring to lemmas that 
will be proved later in this section.
Without loss of generality, suppose that $\pe$ has diameter 2, 
and let $x, y \in \pe$ be two points realizing the diameter of $\pe$. 
Next, fix a longest plane spanning tree $\Topt$ for $\pe$,
and consider a longest edge $ab$ of $\Topt$. 
Note that $ab$ does not 
necessarily realize the diameter of $\pe$, and thus $a$, $b$ and 
$x$, $y$ may differ.
Write 
$\| ab \| = 2d$. By our assumption, we have $d \leq 1$.
Now, if $2d \leq 1/f$, we can use
a technique of Alon, Rajagopalan, and
Suri~\cite{AlonRS95} to show that one of $S_x$ 
or $S_y$ is long enough for the desired
approximation (see \cref{lem:short}).
Thus, we will assume from now on that $2df > 1$.

\begin{figure}[ht]
\center
\includegraphics{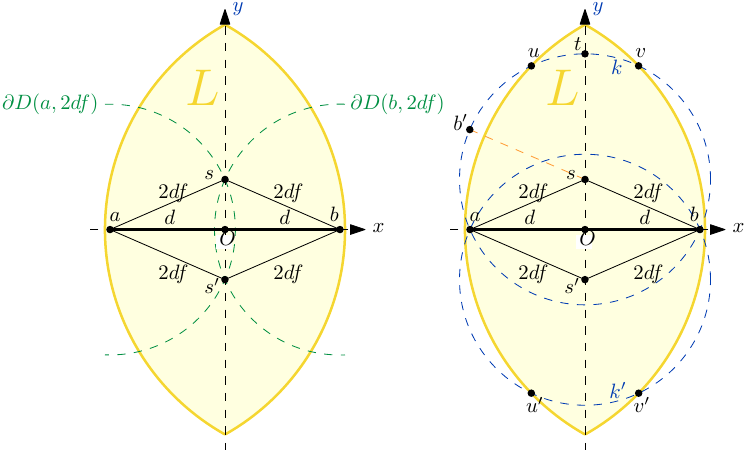}
	\caption{(left) The points $a$ and $b$ define a longest edge of $\Topt$. The points
	$s$ and $s'$ are at distance $2df$ from both $a$ and $b$. The lens $L$ is defined
	by $a$ and $b$ with radius $2$. (right) The points $u$, $v$, $u'$, and $v'$ 
	are at the intersection of the circles with radius $2df$ around $s$ and $s'$ and
	the boundary of $L$.}
\label{fig:lens-first}
\end{figure}

We define several regions in the plane that
we use to analyze $\pe$ and to derive the
desired approximation, refer to Figure~\ref{fig:lens-first}
for an illustration.
First, we choose a coordinate system such that $a$ and $b$ lie on the
$x$-axis and are symmetric with respect to the $y$-axis,
with $a$ on the left. In other words, since $\| ab \| = 2d$, we choose
the coordinates such that $a = (-d, 0)$ and $b = (d, 0)$. 
Since $d \leq 1$ and $2df > 1$, the two circles
$\partial D(a, 2df)$ and $\partial D(b, 2df)$
have exactly two interesction points.
By symmetry, these intersection points lie on the $y$-axis, one above
and one below the $x$-axis.
We let $s$ be the intersection point that is above the $x$-axis,
and $s'$ the intersection point below the $x$-axis.
By definition, $s$ and $s'$ 
have distance $2df$ from both $a$ and $b$. 

Now, we define the \emph{lens}
$L = D(a, 2) \cap D(b, 2)$.
Since $\pe$ has diameter $2$, we have
$\pe \subset L$, and since $2df \leq 2$ (by our choice of $f$),
we have $s, s' \in L$.
Next, consider the circles 
$k = \partial D(s, 2df)$ and 
$k' = \partial D(s', 2df)$.
We argue that $k$ intersects the boundary $\partial L$
of the lens $L$.
Indeed, the point $b'$ that is symmetric to $b$ with
respect to $s$ lies outside $L$, because 
$\|bb'\| = 2 \cdot 2df > 2$. 
On the other hand, the point $t$ on $k$ that is vertically
above $s$ on the $y$-axis lies inside $L$, because it has $y$-coordinate
\begin{align*}
  \sqrt{4d^2f^2 - d^2} + 2df &\leq
  \sqrt{4f^2 - 1}  + 2f\\ 
  &\leq 
  \sqrt{4\cdot 0.55^2 - 1}  + 2\cdot 0.55 \\
  &\leq 1.56 \leq \sqrt{3} \leq 
  \sqrt{4 - d^2},
\end{align*}
as $d \leq 1$ and by our choice of $f \leq 0.55$.
Thus, along the boundary of $k$, there has to 
be exactly one intersction point between $k$ and
$\partial L$ that lies clockwise between $b'$ and $t$.
We call this intersection point $u$.
Symmetrically, we define
$v$ as the analogous intersection between $k$ and $\partial L$
in the upper right quadrant, and $u'$, $v'$ as the analgous
intersections between $k'$ and $\partial L$ in the halfplane
below the $x$-axis.
Now, we claim that $u$ lies to the left of the
$y$-axis and to the right of $a$. Indeed,
we argued that $u$ lies between the points $b'$ and
$t$, as defined above. Since $t$ lies on the $y$-axis,
the first part of the claim follows. 
Furthermore, 
since $b'$ is symmetric to $b$ with respect to $s$, which
lies on the $y$-axis, and since $a$ is symmetric to $b$
with respect to the $y$-axis, it follows that $b'$ lies vertically
above $a$, and hence $u$ also lies to the right of $a$.
By symmetry, analogous statements hold for $u'$, $v$, and
$v'$, and we have
\begin{equation}
\label{equ:au0vb}
a_x \leq u_x = u'_x \leq 0 \leq v_x = v'_x  \leq b_x.
\end{equation}

We use the points $u$, $v$, $u'$ and $v'$ to partition
the lens into two regions, the \emph{far region} and
the \emph{truncated lens}. 
The truncated lens is the region of $L$ that lies
below the arc of 
$k$ between $u$ and $v$ in clockwise direction and above
the arc of $k'$ between $u'$ and $v'$ 
in counter-clockwise direction.
The remaining parts of $L$ constitute the far region, 
see \cref{fig:approx-factor}.

\begin{figure}[ht]
\center
\includegraphics{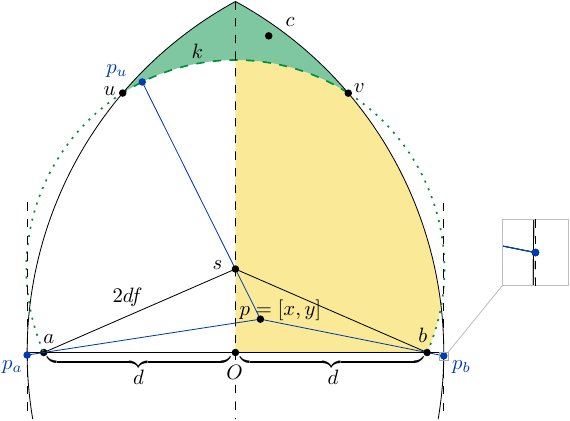}
\caption{The lens $L$ is split into the far region (green) and the 
truncated lens.}
\label{fig:approx-factor}
\end{figure}

In \cref{lem:far}, we show that for any point $c$ in 
the far region, the triangle $abc$ has acute angles and 
circumradius $R \geq 2df$. 
This allows us to argue that if $\pe$ contains a point $c$ 
in the far region, then one of the three stars 
$S_a$, $S_b$, or $S_c$ is long enough for the desired
approximation (\cref{lem:far}). 

It remains to consider the situation that
$\pe$ lies in the truncated lens. 
We claim that in this case, one of the 
trees $S_a$, $T_{a,b}$, $T_{b,a}$, or $S_b$ 
yields the desired approximation.  
For this, we proceed as in the proof of \cref{thm:approximation-flat}.
Our goal is to show
\begin{equation}
\max \{|S_a|,|T_{a,b}|, |T_{b,a}|, |S_b|\} \geq f \cdot \OPT.
\end{equation}
To do this, we fix a parameter $\beta \in (0, 1/2)$, and we 
define a probability distribution $\mathcal{D}_\beta$ on
$\{S_a, T_{a,b}, T_{b,a}, S_b\}$ by
\[
\Pr_{\mathcal{D}_\beta}[S_a] = 
\Pr_{\mathcal{D}_\beta}[S_b] = 1/2 - \beta \text{ and  }
\Pr_{\mathcal{D}_\beta}[T_{a,b}] = 
\Pr_{\mathcal{D}_\beta}[T_{b,a}] = \beta.
\]
Now it suffices to show that there exists a choice of $\beta$ such
that
\begin{equation}
\mathbf{E}_{T \sim \mathcal{D}_\beta}\left[|T|\right] \geq f \cdot \OPT.
\end{equation}
For this, we consider the edges individually and use
linearity of expectation. For $p \in \pe \setminus \{a\}$ and
$T \in \{S_a,T_{a,b}, T_{b,a}, S_b, \Topt\}$, let
$\ell_T(p)$ be the length of the first edge on the path from
$p$ to $a$ in $T$.
Then, we have
\begin{equation}
\mathbf{E}_{T \sim \mathcal{D}_\beta}\left[|T|\right]=
\mathbf{E}_{T \sim \mathcal{D}_\beta}
\left[\sum_{p \in \pe \setminus \{a\}} \ell_T(p)\right]
= \sum_{p \in \pe \setminus \{a\}}
\mathbf{E}_{T \sim \mathcal{D}_\beta} \left[ \ell_T(p)\right].
\end{equation}
To finish the argument, we show that
that for any $p \in \pe \setminus \{a\}$, we have 
\begin{equation} \label{eqn:key}
\mathbf{E}_{T \sim \mathcal{D}_\beta} \left[ \ell_T(p)\right] \geq f \cdot
\ell_{\Topt}(p).
\end{equation}
In contrast to the proof of \cref{thm:approximation-flat}, 
this now requires much more work. Below, we will see that 
$\beta \doteq 0.1604$ gives the desired result. 
Then, we get
\begin{equation} 
\sum_{p \in \pe \setminus \{a\}}
\mathbf{E}_{T \sim \mathcal{D}_\beta} \left[ \ell_T(p)\right] \geq 
\sum_{p \in \pe \setminus \{a\}}
f \cdot  \ell_{\Topt}(p)
=  f \cdot |\Topt| =  f \cdot \OPT.
\end{equation}
Thus at least one of the four trees $S_a$, $T_{a,b}$, $T_{b,a}$, $S_b$ 
has total length at least $f \cdot \OPT$.

We outline the proof of (\ref{eqn:key}), again referring
to lemmas that will be stated and proved below.
If $p = b$, then (\ref{eqn:key}) is immediate,
because the five trees
$S_a$, $T_{a,b}$, $T_{b,a}$, $S_b$, $\Topt$ all contain
the edge $ab$.
Thus, let $p \in \pe \setminus \{a, b\}$, and assume without
loss of generality that $p$ is in the upper right
quadrant, i.e., that $p = (x, y)$ with $x, y \geq 0$. 
First, we establish an upper bound on 
$\ell_{\Topt}(p)$, by estimating the length of a longest possible
line segment that (i) has $p$ as an endpoint, (ii) lies in $L$,
and (iii) does not cross the edge $ab$.
For this, we identify three possible candidates for $p_a$, 
$p_b$, and $p_u$ of the other endpoint
(see \cref{fig:approx-factor} 
for an illustration).
The first two points $p_a$ and $p_b$ provide upper bounds
for the case that the other endpoint of the longest possible edge
lies below the $x$-axis:
the point $p_a$ is the point with $x$-coordinate $-(2 - d)$ on the ray $pa$; 
and for $p_b$, there are two cases: if $x < d$, 
then $p_b$ is the point with $x$-coordinate $2 - d$ on the 
ray $pb$; if $x \geq d$, the ray $pb$ does not intersect the vertical line 
$x = 2 - d$, and we set $p_b = b$.
The third point $p_u$ is the furthest point from $p$ on the part of 
the boundary of the far region that is on the circle 
$k = \partial D(s, 2df)$. 
The proof now proceeds in the following steps:
\begin{enumerate}
\item In \cref{lem:topt},
we establish the upper bound 
\begin{equation}
\label{equ:loptupper}
\ell_{\Topt}(p) \leq \min\big\{2d,\, \max\{\| pp_a\|, 
\| pp_b\|, \| pp_u\| \}\bigr\}, 
\end{equation}
estimating the length of the longest possible line segment in the lens $L$
with endpoint $p$ that does not cross $ab$, and using the fact
that all edges in $\Topt$ have length at most $2d$.
Thus, it suffices to compare 
$\mathbf{E}_{T \sim \mathcal{D}_\beta} \left[ \ell_T(p)\right]$
with $\| pp_a\|$, $\| pp_b\|$, and $\| pp_u\|$.
\item In \cref{lem:right}, we show that $\|pp_a\|$ can
never be smaller than $\|pp_a\|$, and thus the term 
$\| pp_b\| $ in (\ref{equ:loptupper})
can be dropped.
That is, it suffices to compare 
$\mathbf{E}_{T \sim \mathcal{D}_\beta} \left[ \ell_T(p)\right]$
with $\| pp_a\|$ and $\| pp_u\|$.
\item In \cref{lem:ab-construction}, we establish a lower bound on 
$\mathbf{E}_{T \sim \mathcal{D}_\beta} \left[ \ell_T(p)\right]$,
by analyzing how $p$ is connected in the trees $S_a$, $T_{a,b}$,
$T_{b,a}$, and $S_b$.
\item In \cref{lem:left,lem:up}, we relate the lower bound 
from \cref{lem:ab-construction} to 
$\|pp_a\|$ and $\|p p_u\|$.
In this way, we identify constraints on $\beta$ such that
\begin{align*}
\mathbf{E}_{T \sim \mathcal{D}_\beta} \left[ \ell_T(p)\right] &\geq 
f \cdot \min\{2d,\| pp_a\|\} \text{ and }\\
\mathbf{E}_{T \sim \mathcal{D}_\beta} \left[ \ell_T(p)\right] &\geq 
f \cdot \min\{2d,\| pp_u\|\}.
\end{align*}
\end{enumerate}
It remains to find a $\beta \in (0, 1/2)$ that satisfies 
all constraints from \cref{lem:left,lem:up}. 
The most stringent constraints turn out to be those from \cref{lem:up}, 
stating that
\begin{equation}\label{eqn:f-constraints}
   \frac{2f-1}{2\sqrt{5-8f}-1} \leq \beta \leq 1-f\sqrt{4f^2-1}-2f^2.
\end{equation}
In \cref{lem:algebra}, we use straightforward algebra 
to show that 
our choice of $f$ makes the left-hand side and the right-hand side of 
(\ref{eqn:f-constraints}) equal. Hence, 
we set $\beta\doteq 0.1604$, their joint value. 
(The value $f = 5/8$ also makes both sides of 
(\ref{eqn:f-constraints}) equal, but it would give
$\beta < 0$.) 

To summarize, for the approximation guarantee $f \doteq 0.5467$ 
given in the theorem, the value
$\beta =1 -f\sqrt{4f^2-1}-2f^2 \doteq 0.1604$, and 
for every point $p \in \pe \setminus \{a\}$, we have
\begin{align*}
\mathbf{E}_{T \sim \mathcal{D}_\beta} \left[ \ell_T(p)\right] &\geq 
f \cdot \max\big\{ \min\{2d,\| pp_a\|\},\, 
\min\{2d,\| pp_u\| \}\big\} \\
&\geq f\cdot \min\big\{ 2d,\, \max\{\|pp_a\|, \| pp_u\|\}\big\} \\
&\geq f \cdot \ell_{\Topt}(p),
\end{align*}
and the result follows.
\end{proof}

It remains to prove Lemmas~\ref{lem:short} to \ref{lem:algebra}.
Their statements rely on the notation 
introduced in the proof outline of \cref{thm:approximation}, so 
we recommend to first consult the paragraphs above.

\begin{lemma}\label{lem:short} 
Let $x, y \in \pe$ be two
points that realize the diameter of $\pe$. Suppose that $\|xy\| = 2$ and
that all edges of the optimal tree $\Topt$ have length at most $1/f$. 
Then, we have $\max\{|S_x|,|S_y|\} \geq f \cdot \OPT$.
\end{lemma}

\begin{proof} 
By the triangle inequality, for any point $p \in \pe$, 
we have $\| xp\| +\| yp\| \geq \| xy\| = 2$. Hence,
the total length of the stars $S_x$ and $S_y$ is
\[
|S_x| + |S_y| = 
\sum_{p \in \pe} \left( \| xp\| + \| yp\|\right) \geq 
n \cdot \| xy\| = 2n.
\]
On the other hand, since each of the $n - 1$ edges in $\Topt$ has 
length at most $1/f$, we get
\[
\max\{|S_x|,|S_y|\}
\geq 
\frac{|S_x| + |S_y|}{2}
\geq 
n
\geq
f \cdot (n - 1) \cdot \frac{1}{f}
\geq 
f \cdot \OPT,
\]
and we are done.
\end{proof}

\begin{lemma}\label{lem:far} 
Let $ab$ (with $\| ab\| = 2d$) be a longest edge of $\Topt$. 
If $\pe$ contains a point $c$ in the far region, then 
$\max\{|S_a|,|S_b|, |S_c|\} \geq f \cdot \OPT$.
\end{lemma}

\begin{proof} 
First, we argue that the triangle  $abc$ 
is acute:
by (\ref{equ:au0vb}), the far
region lies vertically between $a$ and $b$, and
thus 
the angles at $a$ and $b$
in the triangle $abc$ are at most $\pi/2$.
Furthermore, 
the upper half of the disk $D((0,0), d)$ 
is completely contained in $k$ (\cref{fig:green}). 
Thus, since $c$ lies vertically above $k$, and hence vertically above 
$D((0,0), d)$,
Thales's theorem implies 
that the angle at $c$ is also at most $\pi/2$. It follows that
all three angles are at most $\pi/2$ and 
the triangle $abc$ is acute.

Next, we argue that the circumradius
$R$ of the triangle $abc$ statisfies $R \geq 2df$:
Indeed, since the center of the circumcircle of a triangle
lies at the intersection of the perpendicular bisectors 
of its sides,
the center of the circumcircle of
$abc$ lies on the $y$-axis.
Consider the point $c'$ that lies vertically below
$c$ on $k$. The triangle $abc'$ is acute and has
$k$ as its circumcircle, i.e., its circumradius 
is $2df$. Now, if we move $c'$ vertically towards $c$,
the intersection of the perpendicular bisector of $ac$ with the $y$-axis
moves upward, and hence its distance to $a$ increases.
Thus, thus circumradius $R$ of $abc$ is as least as large as
the circumradius  of $abc'$, and we have
$R \geq 2df$.

Now, set 
$\pe_0 = \pe \setminus \{a, b, c\}$, and let
$g = (1/|\pe_0|) \cdot \sum_{p \in \pe_0} p$ be the corresponding 
center of mass.\footnote{Note that $\pe_0$t may have
points below and above the $x$-axis, but that does not affect the argument.}
By the definition of $g$, 
we have that 
$\sum_{p\in\pe_0} \ds{vp} = |\pe_0|\cdot \ds{vg}$, for every 
point $v \in \R^2$, where $\ds{ab}$ denotes the two-dimensional vector
that is defined by $v$ and $p$.

\begin{figure}[ht]
\center
\includegraphics[page=2]{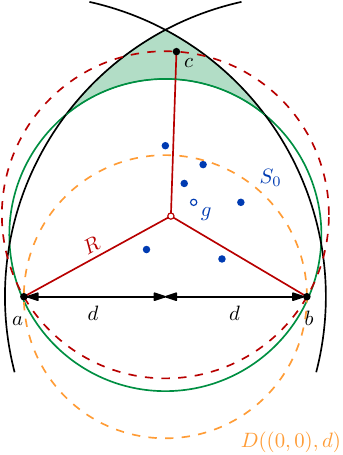}
\caption{We can chose $v=a$.
The orange disk does not intersect the far region above the $x$-axis.
The circumradius of the triangle $abc$ is marked red.}
\label{fig:green}
\end{figure}

By the triangle inequality, this imples that for every $p \in \R^2$,
we have
\[
\sum_{p\in\pe_0} \|vp\| \geq |\pe_0| \cdot \|vg\|
= (n - 3)\, \|vg\|.
\]
Since the triangle $abc$ is acute, Thales's theorem implies that the 
circumcenter $q$ of $abc$ lies in its
interior. Thus, if we consider the rays $qa$, $qb$, and $qc$ from the circumcenter
$q$ through $a$, $b$, and $c$, the angle between two consecutive rays is as most $\pi$,
and for each point $v \in \{a, b, c\}$, the disk $D(v, R)$ does not intersect
the region spanned by the rays through from $q$ through the other two points.
Thus, since the center of mass $g$ lies in one of these regions, 
there exists a vertex $v \in \{a,b,c\}$ 
such that $\| vg\| \geq R$.
The acuteness of $abc$ also implies that
$\|va\| + \|vb\| + \|vc\| \geq 2R$. Thus, 
we get 
\begin{align*}
	|S_v| &= 
\sum_{p \in \pe} \|vp\|\\ 
	&= 
\sum_{p \in \pe_0} \|vp\| + 
	\|va\| + \|vb\| + \|vc\| \\
	& \geq 
(n - 3) \cdot  \|vg\| + 2R\\
	&\geq (n - 3) \cdot R + 2R\\
	&= (n - 1) \cdot R\\
	&\geq (n - 1 )\cdot 2df \\
	&\geq f \cdot \OPT,
\end{align*}
where in the last inequality we used that every edge of $\Topt$ has length at most $2d$.
\end{proof}

Recall that $ab$ is the longest edge of $\Topt$
and that we assumed  $a =(-d, 0)$, $b = (d, 0)$ and $2df > 1$. 
Furthermore, recall that The point $s$ has coordinates $s = (0, \sqrt{(2df)^2 - d^2})$. 
We also noted that the circle $k = D(s, \| sa \|)$ always intersects the lens 
$D(a, 2) \cap D(b, 2)$, and we used this to define the intersection points
$u$ and $v$, as in \cref{fig:lens-first}.

For each point $p = (x,y )$ in the truncated lens with $x,y \geq 0$, recall 
that we have defined the following points:
\begin{itemize}
\item $p_a$: the point on the ray $pa$ whose $x$-coordinate equals $-(2 - d)$;
\item $p_b$: the point on the ray $pb$ whose $x$-coordinate equals $2 - d$, if $x(p) < d$,
	and $p_b = b$, if $x(p) \geq d$; and
\item $p_u$: the point on the arc of $k$ from $u$ to $v$ that is furthest from $p$. Thus,
	if ray $ps$ intersects the arc of $k$ from $u$ to $v$, then $p_u$ is that intersection 
	point, and therwise $p_u = u$.
\end{itemize}

\begin{figure}
\centering
\includegraphics{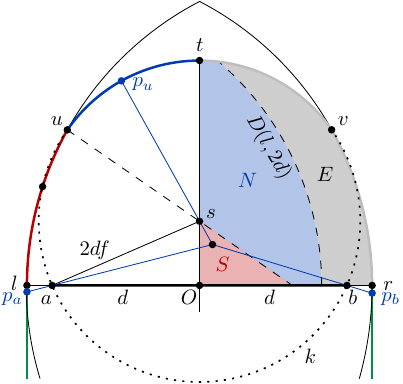}
\caption{
For $p\in E$, we have $\ell_{\Topt}(p) \leq 2d$, for $p \in N$, 
we have $\ell_{\Topt}(p) \leq \max\{\| pp_a\|,\| pp_u\| \}$, 
and for $p \in S$, we have $\ell_{\Topt}(p) \leq \max\{\| pp_a\|,\| pp_u\|\}$.}
\label{fig:topt}
\end{figure}

Now, we show that these three special points $p_a$, $p_b$, and $p_u$ suffice 
to obtain an upper bound for $\ell_{\Topt}(p)$.
\begin{lemma}\label{lem:topt}
For every point $p = (x, y)$ in the truncated lens with $x, y\geq0$, 
we have
\[
 \ell_{\Topt}(p) \leq \min\left\{2d, \max\{\| pp_a\|,\| pp_b\|,\| pp_u\|\}\right\}.
\]
\end{lemma}
\begin{proof}
As $p$ lies in the upper right quadrant, we have to consider only 
the truncated lens in this quadrant. 
Let $l$ and $r$ be the left- and rightmost points of $D(a,2) \cap D(b,2)$, 
i.e., let $r = (2 - d, 0)$ and $l = (d - 2, 0)$.

We further subdivide the truncated lens into regions; see \cref{fig:topt} for an illustration: 
(i) the region $E$ lies inside the truncated lens, but outside of $D(l, 2d)$ by $E$; 
(ii) the region $N$ lies in the intersection of the truncated lens and $D(l, 2d)$, and
above the line through $u$ and $s$; and  (iii) the region
$S$ lies in the intersection of the truncated lens and $D(l, 2d)$,
and below the line though $u$ and $s$.

First, consider the case that $p \in E$. Then, we have $\| pp_a\| \geq \| pl\| \geq 2d$. 
Hence, it follows that
\[
	\max\{\| pp_a\|,\| pp_b\|,\| pp_u\| \leq 2d, 
\]
and thus
\[
	\min \left\{2d, \max \{\| pp_a\|,\| pp_b\|,\| pp_u\|\} \right\} = 2d.
\]
Since by our definition of $d$, we have
$\ell_{\Topt}(p) \leq 2d$ for every $p \in \pe$, 
the claim now follows for the case $p \in E$.

\begin{figure}[ht]
\center
\includegraphics[page=2]{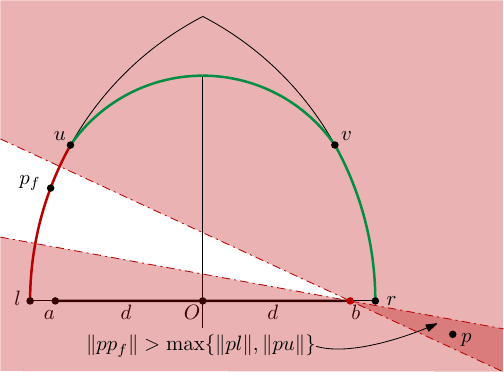}
\caption{The point $p_f$ lies on the arc $ul$. The set of points $p$ 
	such that $\| pp_f\|> \max\{ \| pl\|,\| pu\| \}$ 
	forms a convex wedge with vertex $b$ fully contained in the fourth quadrant.
}
\label{fig:topt_case5}
\end{figure}

Next, consider the case that $p \in N\cup S$. Now, since
$p$ must be to the left of $b$, it follows that $x(p_b)=2 - d$.
Let $p_f$ be the point within the truncated lens that is furthest from $p$ and can 
be connected to $p$ by a line segment that does not intersect $ab$.
We assume that there is no point from $\pe$ in the far region, so we must
have $\ell_{\Topt}(p) \leq \|pp_f\|.$
Since the truncated lens is bounded, $p_f$ lies on the boundary of the truncated lens.
Since $\ell_{\Topt}(p) \leq 2d$ by our definition of $d$, it 
suffices to show that $\| pp_f\| \leq \max\{\| pp_a\|,\| pp_b\|,\| pp_u\|\}$.
We distinguish four cases, depending on the quadrant containing $p_f$:
\begin{enumerate}
\item  $p_f$ lies in the upper right quadrant: consider the reflection $p_f'$ 
	of $p_f$ about the $y$-axis. Since $x \geq 0$, 
	we have $\| pp_f'\| \geq \| pp_f\|$, and the inequality is strict for 
	$x \geq 0$. Thus, if $x > 0$, this case cannot occur, and if $x = 0$, 
	it is handled in the second case
\item 
$p_f$ lies in the left quadrant: 
let $t$ be the top-most point of the truncated lens.
		We distinguish two subcases:
\begin{enumerate}
	\item $p_f$ lies on the arc $tu$ (\cref{fig:topt}):
if $p$ lies in $N$, then for any $q$ on the arc $tu$, we have 
$\| pq\| \leq \| pu\| = \| pp_u\|$, thus also $\| pp_f\| \leq \| pp_u\|$.
If $p$ lies in $S$, then, by the triangle inequality, 
we get $\| pp_f\| \leq \| ps\| + \| sp_f\|=\| ps\| + 2df = \| pp_u\|$. Thus, in
either situation, we have $\|pp_f\| \leq \|p p_u\|$.
\item $p_f$ lies on the arc $ul$: we claim that now, 
it holds that $\| pp_f\| \leq \max\{ \| pl\|,\| pu\| \}$. Indeed, 
since $l$, $p_f$, and $u$ all lie on $\partial D(b, 2)$, 
the perpendicular bisectors of the segments $p_fl$ and $p_fu$ intersect at $b$, and 
thus the points $q$ for which $\| qp_f\| \leq \max\{ \| ql\|,\| qu\| \}$ 
lie in a convex cone with vertex $b$ that is fully contained in the lower right quadrant, 
see \cref{fig:topt_case5}. Now, since $\| pl\| \leq \| pp_a\|$ and 
$\| pu\| \leq \| pp_u\|$, we get $\| pp_f\| \leq \max\{\| pp_a\|, \| pp_u\|\}$ as desired.
\end{enumerate}
\item  $p_f$ lies in the lower left quadrant: then, we have $\| pp_f\| \leq \| pp_a\|$,
	by the definition of $p_a$.
\item  $p_f$ lies in the lower right quadrant: then, we have  $\| pp_f\|\leq \| pp_b\|$,
	by the definition of $p_b$.\qedhere
\end{enumerate}
\end{proof}

Now we show, that if $p$ lies in the first quadrant, the point $p_b$ can be ignored.
\begin{lemma}\label{lem:right} 
  For every point $p = (x, y)$ with $x, y \geq 0$ in the truncated lens,
  we have that if $\| pp_a\| \leq 2d$, then $\| pp_b\| \leq \| pp_a\|$.
\end{lemma}
\begin{proof} 
We present a direct analytic proof 
(we are not aware of an argument that uses
only elementary geometry).
First, if $x = 0$, the claim  follows from symmetry, since then we have 
$\| pp_b\| = \| pp_a\|$. 
Thus, we consider only the case $x > 0$.
Since $\| pp_a\| \leq 2d$, we must have $x < d$, and therefore $p_b$ has $x$-coordinate $2 - d$.
See \cref{fig:right}.

\begin{figure}[ht]
\center
\includegraphics[scale=0.9]{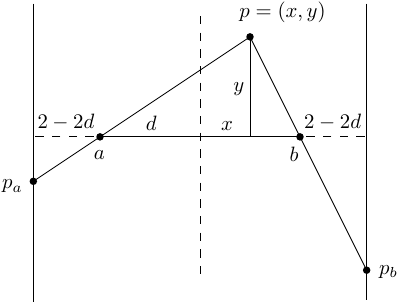}
\caption{The situation for \cref{lem:right}.
}
\label{fig:right}
\end{figure}

Using similar triangles, we have 
$\| pp_a\|=\| pa\| \cdot \frac{2 - d + x}{d + x}$,
and using the Pythagorean theorem, we have $\| pa\|^2 = y^2 + (d + x)^2$.
Thus, we have
\begin{equation}
	\label{equ:ppa}
\| p p_a \| ^2 = (y^2 + (d + x)^2) \cdot \frac{(2 - d + x)^2}{(d + x)^2}.
\end{equation}
Similarly, using 
$\| pp_b\|=\| pb\|\cdot \frac{2 -d +x}{d - x}$ and
$\| pb\|^2 = y^2 + (d - x)^2$,
we express $\| pp_b\|^2$ as
	\[
	\| pp_b\|^2 = (y^2 + (d - x)^2) \cdot \frac{(2 - d - x)^2}{(d - x)^2}.
\]
Thus, our goal $\| pp_b\|^2 \leq \| pp_a\|^2$ 
can be stated as  
\[
(y^2 + (d - x)^2)\cdot\frac{(2 - d - x)^2}{(d - x)^2}   \leq   
	(y^2+(d+x)^2)\cdot\frac{(2-d+x)^2}{(d+x)^2}.
\]
Separating the terms involving $y^2$ from the rest, this becomes
	\begin{equation}
		\label{equ:ysquare}
y^2\cdot \frac{(2-d-x)^2(d+x)^2 - (2-d+x)^2(d-x)^2}{(d+x)^2(d-x)^2} \leq (2-d+x)^2 - (2-d-x)^2.  
	\end{equation}
	For the right-hand side of (\ref{equ:ysquare}), we have
\begin{align*}
	(2 - d + x)^2 - (2 -d -x)^2 &= (2 - d)^2 + x^2 + 2(2 - d)x - (2 - d)^2 - x^2 + 2(2 - d) x\\
	&= 4(2 - d) x.
\end{align*}
	For the numerator on the left-hand side of (\ref{equ:ysquare}), we get
\begin{align*}
	&(2-d-x)^2(d+x)^2 - (2-d+x)^2(d-x)^2 \\
	&=	((2-d)^2-2(2-d)x + x^2)(d^2+2xd + x^2) - ((2-d)^2+2(2-d)x) + x^2)(d^2 - 2xd + x^2)\\
	&= 2(2 - d)^2 2xd - 4(2 - d)xd^2 - 4(2 - d)x^3 + 2x^2 2xd\\
	&= 4(2 - d)^2 xd - 4(2 - d)xd^2 - 4(2 - d)x^3 + 4x^3d\\
	&= 4x \left[(2 - d)^2d - (2 - d)d^2 - (2 - d)x^2 + x^2d\right]\\
	&= 4x \left[4d -4d^2 + d^3 - 2d^2 + d^3 - 2x^2 + dx^2  + dx^2\right]\\
	&= 4x \left[4d -4d^2 + d^3 - 2d^2 + d^3 - 2x^2 + 2dx^2\right]\\
	&= 4x \left[4d - 2d^2  - 2x^2 - 4d^2 + 2d^3 + 2dx^2\right]\\
	&= 4x \cdot 2 (1-d)(2d - d^2 - x^2).
\end{align*}
Thus, by plugging back into (\ref{equ:ysquare}) and dividing by $4x > 0$,
the goal
$\| pp_b\|^2 \leq \| pp_a\|^2$ becomes
\begin{equation}
	\label{equ:goal}
	y^2\cdot \frac{2(1-d)(2d-d^2-x^2)}{(d+x)^2(d-x)^2} \leq 2-d.
\end{equation}
Using (\ref{equ:ppa}), the assumption $\| pp_a\|^2 \leq (2d)^2$ can be equivalently rewritten as
\[
(y^2 + (d + x)^2)\cdot\frac{(2 - d + x)^2}{(d + x)^2} \leq 4d^2. 
\]
Solving for $y^2$, this gives
\begin{equation}
\label{equ:ppa2}
y^2 \leq \frac{(4d^2-(2-d+x)^2)(d+x)^2}{(2-d+x)^2}. 
\end{equation}
We plug in the upper bound on $y^2$ from (\ref{equ:ppa2}) into
(\ref{equ:goal}), 
cancel the term $(d + x)^2$, and clear the denominators.
This leaves us with proving
\[ 
(4d^2-(2-d+x)^2) \cdot 2(1-d)(2d-d^2-x^2) \le (2-d+x)^2 \cdot (2-d) (d-x)^2
\]
which, upon expanding the parentheses and collecting the terms, becomes 
\begin{equation}\label{eqn:pol}
 0\le (16 d - 32 d^2 + 8 d^3 + 16 d^4 - 7 d^5) + x^2(- 8 d+ 16 d^2 - 10 d^3) + d x^4.
\end{equation}
For any fixed $d > 0$, the right-hand side $Q(d,x^2)$ is 
a quadratic function of $x^2$ with positive coefficient $d > 0$ 
for the leading term $(x^2)^2$. Hence, the minimum of $Q(d,x^2)$ is attained when
\[ x^2=\frac{8d-16d^2+10d^3}{2d} = 4-8d+5d^2.
\]
Plugging $x^2=4-8d+5d^2$ into (\ref{eqn:pol}) and 
expanding the parentheses for one last time, we are left to prove
\[
0 \leq 32d^2 - 96d^3 + 96d^4 - 32d^5 = 32d^2 ( 1 - 3d + 3d^2 - d^3) = 
32d^2(1 - d)^3,
\]
which is true since $d \leq 1$.
\end{proof}

Now we give a general lower bound on  $\mathbf{E}_{T \sim \mathcal{D}_\beta}[\ell_T(p)]$ 
that we will then use in \cref{lem:up,lem:left}.
\begin{lemma}\label{lem:ab-construction} 
  Let $p = (x, y)$ be any point in the plane with $x,y \geq 0$ 
  and let $\beta \in (0, 1/2)$ be a real number. Then
  \[
    \mathbf{E}_{T \sim \mathcal{D}_\beta}[\ell_T(p)]
    \geq \frac{d\cdot (1-\beta)+x\cdot 2\beta}{d+x}\cdot \| pa\|.
\]
\end{lemma}
\begin{proof} Expanding the definition, we have
\[
\mathbf{E}_{T \sim \mathcal{D}_\beta}[\ell_T(p)]
\geq (1/2 - \beta)\cdot \| pa\| + \beta\cdot \| pa\| + \beta\cdot x + (1/2-\beta)\cdot \| pb\|.
\]

\begin{figure}[ht]
\center
\includegraphics{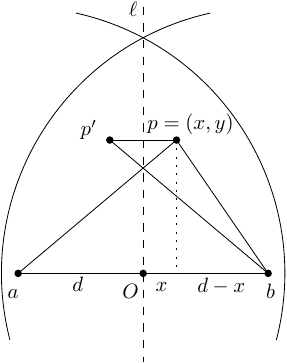}
\caption{Mirroring $p$ along the $y$-axis in \cref{lem:ab-construction}.
}
\label{fig:ab}
\end{figure}
Let $p' = (-x, y)$ be the reflection of $p$ at the $y$-axis (see \cref{fig:ab}).
By the triangle inequality, $\| p'p\|+\| pb\|\ge \| p'b\| =\| pa\|$. This leads to
	$x = \frac{1}{2} \| p'p\| \geq \frac{1}{2} (\| pa\| - \| pb\|)$,
and we obtain
\[
\mathbf{E}_{T \sim \mathcal{D}_\beta}[\ell_T(p)]
	\geq \frac{1}{2}\cdot \| pa\| + \frac12\beta\cdot \| pa\| + 
	\left(\frac12-\frac32\beta\right)\cdot \| pb\|.
\]
Next, we claim that $\| pb\| \geq \frac{d-x}{d+x} \cdot \| pa\|$: Indeed, 
upon squaring, using the Pythagorean theorem and clearing the denominators 
this becomes $y^2\cdot 4dx\ge 0$ which is true. 
Using this bound on the term containing $\| pb\|$, we finally get the desired
\[
\mathbf{E}_{T \sim \mathcal{D}_\beta}[\ell_T(p)]
	\geq \frac{(1+\beta)(d+x) + (1-3\beta)(d-x)}{2(d+x)} \cdot \| pa\| = \frac{(1-\beta)\cdot d + 2\beta\cdot x}{d+x} \cdot \| pa\|. \qedhere
\]
\end{proof}

\begin{lemma}\label{lem:left} Let $p=(x,y)$ be any point in the truncated lens with $x,y\ge 0$. Then, if
\[
\frac{2f-1}{5-8f} \le \beta \le \frac12\cdot f,
\]
we have
$\mathbf{E}_{T \sim \mathcal{D}_\beta}[\ell_T(p)]
\geq f\cdot \min\{2d,\| pp_a\|\}$.
\end{lemma}
\begin{proof}
It suffices to show that:
\begin{enumerate}
\item If $x\ge 3d-2$, then $\mathbf{E}_{T \sim \mathcal{D}_\beta}[\ell_T(p)] \geq f\cdot 2d$.
\item If $x\le 3d-2$, then $\mathbf{E}_{T \sim \mathcal{D}_\beta}[\ell_T(p)] \geq f\cdot \| pp_a\|$.
\end{enumerate}
We consider those two cases independently.
\begin{enumerate}
\item\label{item:first} Using Lemma~\ref{lem:ab-construction} and the inequalities $\| pa\|\ge d+x$ and $x\ge 3d-2$, we rewrite
\[\mathbf{E}_{T \sim \mathcal{D}_\beta}[\ell_T(p)] \geq \frac{d\cdot (1-\beta)+x\cdot 2\beta}{d+x}\cdot \| pa\| \ge d-d\beta +(3d-2)\cdot2\beta = \beta(5d-4)+d.
\]
Hence it suffices to prove $\beta\cdot (5d-4)\ge d(2f-1)$.
Using the lower bound on $\beta$ and  $4\le 8df$ we get
\[\beta\cdot (5d-4) \ge \beta\cdot (5d-8df) \ge\frac{2f-1}{5-8f}\cdot d\cdot (5-8f) = d(2f-1)\]
as desired. Note that $2f-1$ and $5-8f$ are both positive.

\item We have $\| pp_a\|=\| pa\|\cdot \frac{2-d+x}{d+x}$. Using Lemma~\ref{lem:ab-construction} it suffices to prove
\begin{align*}
d(1-\beta)+x\cdot 2\beta &\ge f(2-d+x) \\
d(1-\beta) -f(2-d) &\ge x(f-2\beta).
\end{align*}
Since $f\ge 2\beta$ by assumption, the right-hand side is increasing in $x$ and we can plug in $3d-2$ for $x$.
This leaves us with proving the inequality 
\begin{align*}
d(1-\beta) -f(2-d) &\ge (3d-2)(f-2\beta)\\
\beta\cdot (5d-4)  &\ge d(2f-1),
\end{align*}
which is the same inequality as in the first case.\qedhere
\end{enumerate}
\end{proof}

\begin{lemma}\label{lem:up} 
Let $p=(x,y)$ be any point in the truncated lens with $x,y\ge 0$. 
Suppose that $\beta< \frac{151}{304}\cdot f$, that $\frac12 \le f\le \frac{19}{32}$, and that
\begin{equation*}
\frac{2f-1}{2\sqrt{5-8f}-1} \le \beta \le 1-f\sqrt{4f^2-1}-2f^2.
\end{equation*}
Then 
$\mathbf{E}_{T \sim \mathcal{D}_\beta}[\ell_T(p)] \ge f\cdot \min\{2d,\| pp_u\|\}\label{eq:lemup}$.
\end{lemma}
\begin{proof} 
Because of the lower bound for
$\mathbf{E}_{T \sim \mathcal{D}_\beta}[\ell_T(p)]$ in \cref{lem:ab-construction}, to show the statement it suffices to show that
\begin{equation}
\frac{d\cdot (1-\beta)+x\cdot 2\beta}{d+x}\cdot \| pa\| \geq f\cdot  \min\{2d, \| pp_u \|\}. \label{eq:lemup2}
\end{equation}
We will consider the two cases $y\leq y(u)$ and $y > y(u)$ separately.

\paragraph*{Case 1: $y\leq y(u)$}
Define $\lambda =\frac{d\cdot (1-\beta)+x\cdot 2\beta}{d+x} \cdot \| pa \|$.
We have to show that $\lambda \geq f\cdot  \min\{2d, \| pp_u \|\}$. 
Note that $\lambda$ is an increasing function in $y$ because $\| pa \|$ is also increasing in $y$. 
On the other hand, $\| pp_u \|$ is a decreasing function in $y$ for $y \le y(u)$: 
If $p_u\ne u$, then $\| pp_u\|=\| ps\|+2df$ is decreasing in $y$, and if 
$p_u=u$ then $\| pp_u\| = \| pu\|$ decreases when $y$ increases (for $y\le y(u)$). 
It follows that $\min\{2d, \| pp_u \|\}$ is a decreasing function in $y$ for $y\le y(u)$.

As $\min\{2d, \| pp_u \|\}$ is decreasing and $\lambda$ is increasing in $y$, for $y \leq y(u)$, 
to handle our current case it suffices to show (\ref{eq:lemup2}) for $y = 0$.
Now we have $\| pa\| = d + x$, so we can rewrite $\lambda$ as 
	$\lambda_0 =d \cdot (1-\beta) + x \cdot 2 \beta$, which is positive.

Let $q=(q_x,0), q_x\ge0$ be the point on the $x$-axis such that $\| qs\|=2d(1-f)$.
This means that for $p=q$ we have $\| pp_u\|\le \| ps\|+\| sp_u\|=2d$ (see Figure~\ref{fig:up}).
We consider subcases depending on whether $0\leq x \leq q_x$ or $x> q_x$.

\begin{figure}[ht]
\center
\includegraphics{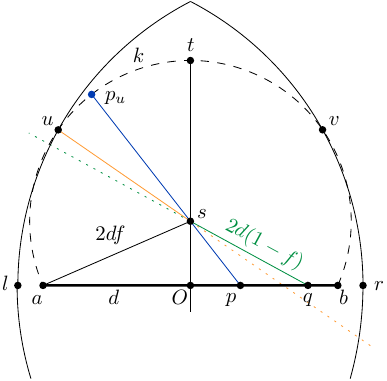}
\caption{Lemma~\ref{lem:up}.
}
\label{fig:up}
\end{figure}

\begin{description}
\item[Case 1a:] $0\leq x \leq q_x$.
We will show that in this case $\lambda_0 \geq f\cdot \| pp_u\|$.
The Pythagorean theorem gives
\[f\cdot \| pp_u\|\le f\cdot\| ps\| + f\cdot \| sp_u\|=f\sqrt{x^2+(2df)^2-d^2} + 2df^2
\]
Therefore, substituting $\lambda_0$, it suffices to show
 \begin{align}
 d \cdot (1-\beta) + x \cdot 2 \beta &\geq f\sqrt{x^2+(2df)^2-d^2} + 2df^2. \label{eq:case1a}
 \end{align}
When $\beta< \frac{151}{304}\cdot f$ and $f\le \frac{19}{32}$, the term $d \cdot (1-\beta) - 2df^2$ is positive and the last inequality is equivalent to 
 \begin{align*}
 0 &\geq f^2 \cdot (x^2+(2df)^2-d^2) - (d \cdot (1-\beta) + x \cdot 2 \beta - 2df^2)^2
 \end{align*}
The right-hand side is a quadratic function of $x$ with positive coefficient $f^2-4\beta^2$ by the leading term, hence it suffices to check the inequality (\ref{eq:case1a}) for $x\in\{0,q_x\}$.

For $x=0$, we need to check that $d(1-\beta)\ge f(\sqrt{d^2(4f^2-1)} + 2df^2)$,
which reduces precisely to the assumption
\begin{equation*}
\beta\le 1- f\sqrt{4f^2-1} - 2f^2.
\end{equation*}

For $x=q_x$, the Pythagorean theorem gives $(q_x)^2=(2d(1-f))^2 + d^2 - (2df)^2 = d^2\cdot (5-8f)$, hence $q_x=d\sqrt{5-8f}$. The point $q$ has been selected so that $\| ps\|+\| sp_u\|=2d$, and therefore
the right side of (\ref{eq:case1a}) is $2df$.
We thus have to verify that
\begin{align}
d(1-\beta)+d\sqrt{5-8f}\cdot 2\beta &\ge f\cdot 2d \label{eq:case1aa}\\
\beta\cdot (2\sqrt{5-8f} -1) &\ge 2f-1.\nonumber
\end{align}
Since $f\le 19/32$, the term in the parentheses on the left-hand side is positive, and after dividing we obtain precisely the assumption. 

\item[Case 1b:] $q_x <x$.
In this case we show that $\lambda_0 \geq f\cdot 2d$.
Since the term $\lambda_0$ is increasing in $x$ and $2d$ is constant, we only need to show that
$\lambda_0 \geq f\cdot 2d$ for $x=q_x$.
However, this was already shown in the previous case; see (\ref{eq:case1aa}).
\end{description}

\paragraph*{Case 2: $y > y(p_u)$}
In this case we have $\| pp_u\| = \| pu\| \leq \| uv\|$. 
Furthermore, because $4df > 2d$, the intersection of the line supporting 
$bs$ with the circle $k$ is outside the lens. This intersection point 
has $x$-coordinate $-d$ because of symmetry with respect to $s$, and since $u$ is 
above it, we have $x(u) \geq -d$ as well as $x \leq x(v) \leq d$. 
So we have $\min\{2d, \| p p_u\|\} = \| pp_u\|$. This means that to show 
(\ref{eq:lemup2}), we have to show
\begin{align*}
\frac{d(1-\beta)+2x\beta}{d+x} \cdot \| pa\| & \geq f\cdot\| pp_u\|.
\end{align*}
As $p$ lies above the horizontal line through $u$, 
we get $\| pa\| \geq \| pp_u\|$, thus it suffices to show
\begin{align*}
\frac{d(1-\beta)+2x\beta}{d+x}  & \geq f\\
d-d\beta + 2x\beta & \geq fd + fx\\
d(1-(\beta+f)) & \geq x(f-2\beta)\\
x & \leq d\cdot \frac{1-(\beta+f)}{f-2\beta} ,
\end{align*}
we we have used that $f-2\beta>0$.
As we know $x\leq d$, this is true for
\begin{align*}
\frac{1-(\beta+f)}{f-2\beta} & \geq 1\\
1-(\beta +f) & \geq f-2\beta\\
\beta &\geq 2f-1,
\end{align*}
where we have again used that $f-2\beta>0$.
The last condition is a looser bound than the left side of the statement of the lemma as $f\geq \frac{1}{2}$ and therefore $2\cdot \sqrt{5-8f}-1 \leq 1$. \qedhere
\end{proof}

\begin{lemma}\label{lem:algebra}
	The positive solutions of
\[
		\frac{2x-1}{2\sqrt{5-8x}-1} = 1-x\sqrt{4x^2-1}-2x^2
\]
	are $x=\tfrac 58$ and the fourth smallest root of
\[
		-80 + 128 x + 504 x^2 - 768 x^3 - 845 x^4 + 1096 x^5 + 256 x^6.
\]
\end{lemma}
\begin{proof}
	We provide a sketch of how to solve it ``by hand''.
	One can also use advanced software for algebraic manipulation.
	Setting the polynomials $q_1(x)=4x^2-1$ and $q_2(x)=5-8x$, and multiplying both sides of the equation by the denominator on the left-hand side, we are left with the equation
	
	\begin{align*}
	 2x-2x^2 &= x \sqrt{q_1(x)} + (2-4x^2) \sqrt{q_2(x)} - 2x\sqrt{q_1(x)q_2(x)}\\
	 2x-2x^2 + 2x\sqrt{q_1(x)q_2(x)} &= x \sqrt{q_1(x)} + (2-4x^2) \sqrt{q_2(x)}.
	\end{align*}
	Squaring both sides, which may introduce additional roots, we get, for some 
	polynomials $q_3(\,),\dots, q_6(\,)$, the equation
	
	\begin{align*}
	 q_3(x) + q_4(x) \sqrt{q_1(x)q_2(x)} &= q_5(x) + q_6(x) \sqrt{q_1(x)q_2(x)}\\
	 q_3(x) -q_5(x) &= \bigl( q_6(x)-q_4(x)\bigr) \sqrt{q_1(x)q_2(x)}.
	\end{align*}
	Squaring both sides gain, which may introduce additional solutions,
	we get the polynomial 
\[		8 (x-\tfrac 58) (-80 + 128 x + 504 x^2 - 768 x^3 - 845 x^4 + 1096 x^5 + 256 x^6) =0
\]
	This polynomial has 7 real roots that can be approximated numerically.
	The smallest three roots of this polynomial are 
	negative ($x\doteq -4.82037$, $x\doteq -0.657898$ and $x\doteq -0.523446$).
	The fourth smallest root, $x\doteq 0.546723$, is a solution to the original equation. 
	The fifth and sixth roots are $x\doteq 0.577526$ and $x\doteq 0.596211$, which
	are not solutions to the original equation.
	The largest root of the polynomial is $x=\tfrac 58 $, 
	which is also a solution to the original equation.	
\end{proof}

\section{Convex and flat convex point sets}\label{sec:caterpillars}

In this section we present two results for convex point sets.
First, we show that if $\pe$ is a convex point set then the longest plane tree is a caterpillar (see~\cref{thm:convex}) and that any caterpillar could be the unique longest plane tree (see~\cref{thm:caterpillars}).
Second, by looking at suitable flat convex sets we prove upper bounds on the approximation factor $\bd(d)$ achieved by the longest plane tree among those with diameter at most $d$.

\subsection{Convex sets and caterpillars}\label{sec:convex}

A tree $C$ is called \emph{caterpillar} if it contains a path $P$ 
such that every node in $C \setminus P$ is adjacent to a node on $P$.
Equivalently, a tree is a caterpillar if its edges can be listed in such an order that every two consecutive edges in the list share an endpoint.

\begin{figure}[ht]
\center
\includegraphics{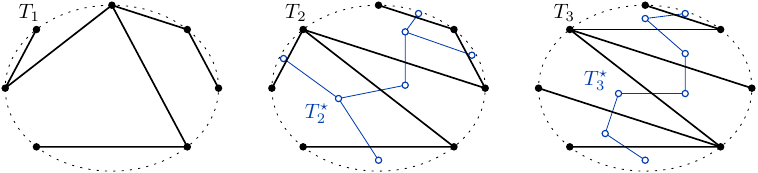}
\caption{Left: A tree $T_1$ that spans $\pe$ but it is not a caterpillar.
Middle: A tree $T_2$ that is a caterpillar but it is not zigzagging.
Right: A tree $T_3$ that is a zigzagging caterpillar -- the dual tree $T_3^\star$ is a path.
}
\label{fig:zigzag}
\end{figure}

Throughout this section we consider trees that span a given convex point set $\pe$.
We say that (a drawing of) such a tree $T$ is a \textit{zigzagging caterpillar} if $T$ is a caterpillar and the dual graph $T^\star$ of $T$ is a path.
Here a \textit{dual graph} $T^\star$ is defined as follows:
Let $C$ be a smooth closed curve passing through all points of $\pe$.
The curve bounds a convex region and the $n-1$ edges of $T$ split that region into $n$ subregions.
The graph $T^\star$ has a node for each such subregion and two nodes are connected if their subregions share an edge of $T$ (see~\cref{fig:zigzag}).

First we prove that the longest plane tree of any convex set is a zigzagging caterpillar.
\thmconvex*
\begin{proof}
Let $\Topt$ be a longest plane tree. We prove that $\Toptd$ 
is a path. Suppose not, and consider a node in $\Toptd$ of degree at least 3. Let $ab$, $bc$, $cd$ 
be three corresponding edges of $\Topt$ (see~\cref{fig:convex}).

\begin{figure}[h]
\center
\includegraphics{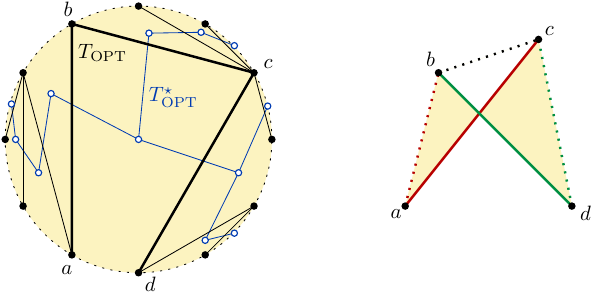}
\caption{Left: A convex point set with a (fake) longest plane tree $\Topt$ (black) and the dual graph $\Toptd$ (blue). Right: If $\Toptd$ has a node of degree 3 or more, a longer tree can be constructed by replacing one dotted edge of $\Topt$ by a longer edge (of the same color).
}
\label{fig:convex}
\end{figure}

As  $abcd$ is a convex quadrilateral, the triangle inequality gives 
$\Vert ab\Vert +\Vert cd\Vert <\Vert ac\Vert +\Vert bd\Vert $, so
$\Vert ab\Vert <\Vert ac\Vert $ or $\Vert cd\Vert <\Vert bd\Vert $ (or both). 
Now, $T_1=\Topt\cup ac\setminus ab$ and 
$T_2=\Topt\cup bd\setminus cd$ are plane trees, and at least one of them 
is longer than $\Topt$, a contradiction. \qedhere
\end{proof}

\def\Cover{\operatorname{Cov}}
Conversely, for each caterpillar $C$ we will construct a convex set $\pe_C$ whose longest plane tree is isomorphic to $C$.
In fact, $\pe_C$ will be not only convex but also a \emph{flat arc}.
Formally, we say that a set $\pe$ of $n=m+1$ points $a_i=(x_i,y_i)$ satisfying $x_1<\dots<x_{m+1}$ is a \textit{flat arc} if it is flat (that is, the absolute values of all $y$-coordinates are negligible) and the points $a_1,\dots,a_{m+1}$ all lie on the convex hull of $\pe$ in this order.

We call the sequence $G(\pe_C)=\{g_i\}_{i=1}^m=\{| x_{i+1}-x_i|  \}_{i=1}^m$ the (horizontal) \textit{gap sequence} of \(\pe_C\).
Lastly, given a tree $T$ spanning a flat arc $\pe$, we define its \textit{cover sequence} $\Cover(T)=\{c_i\}_{i=1}^m$ to be a list storing the number of times each gap is ``covered''. Formally, $c_i$ is the number of edges of $T$ that contain a point whose $x$-coordinate lies within the open interval $(x_i,x_{i+1})$, see~\cref{fig:flat-convex-2}.
Note that the gap sequence and the cover sequence determine the length of the tree because $|T|= \sum_{i=1}^m c_i\cdot g_i$.
Recall here, that in a flat point set, the \(y\)-coordinates can be chosen arbitrarily small and thus we ignore them in our considerations.
\begin{figure}[h]
\center
\includegraphics[scale=0.8]{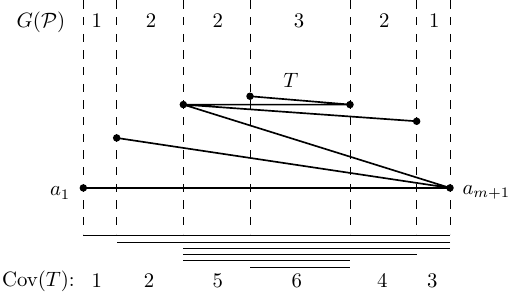}
\caption{Given a flat arc $\pe$ of $n=m+1$ points, we denote by $G(\pe)=\{g_i\}_{i=1}^m$ its gap sequence. Given a tree $T$ spanning $\pe$, we denote by $\Cover(T)=\{c_i\}_{i=1}^m$ its cover sequence. Later we show that the longest plane tree $\Topt$ of a flat arc $\pe$ is a zigzagging caterpillar containing the edge $a_1a_{m+1}$ and that its cover sequence $\Cover(\Topt)$ is a unimodal (single-peaked) permutation of $\{1,2,\dots,m\}$.}
\label{fig:flat-convex-2}
\end{figure}

Before we prove that any caterpillar can be the longest tree of a flat arc, we first show that the cover sequences of zigzagging caterpillars are unimodal (single-peaked) permutations.
\begin{lemma}\label{lem:unimodal}
Consider a flat arc $\pe_C=\{a_1,\dots,a_{m+1}\}$ and a zigzagging caterpillar $T$ containing the edge $a_1a_{m+1}$.
Then the cover sequence $\Cover(T)$ of $T$ is a unimodal permutation of $\{1,2,\dots,m\}$.
\end{lemma}
\begin{proof}
We show this lemma by induction on $m$. The case $m=1$ is clear.
Fix $m\ge 2$. By the definition of a zigzagging caterpillar, the dual graph $T^\star$ of $T$ is a path. Since, by the assumption of the lemma, $a_1a_{m+1}$ is an edge of $T$, either $a_1a_m$ or $a_2a_{m+1}$ is an edge of $T$ too. Without loss of generality assume $a_1a_m$ is an edge of $T$. Then $T\setminus\{a_1a_{m+1}\}$ is a zigzagging caterpillar on $m$ points $a_1,\dots,a_m$ containing the edge $a_1a_m$, hence by induction its cover sequence is a unimodal permutation of $\{1,2,\dots,m-1\}$. Adding the omitted edge $a_1a_{m+1}$ adds 1 to each of the $m-1$ elements and appends a 1 to the list, giving rise to a unimodal permutation of $\{1,2,\dots,m\}$. This completes the proof.
\end{proof}

Now we are ready to prove that any caterpillar, including a path, can be the longest plane tree of some flat arc.
\thmcaterpillars*
\begin{proof}
  We will define a flat arc $\pe_C$ such that the unique longest tree $\Topt$ of $\pe_C$ is isomorphic to $C$.
  Consider a flat arc $\pe=\{a_1,\dots,a_{m+1}\}$, with a yet unspecified gap sequence $\{g_i\}_{i=1}^m$.
Take the caterpillar $C$ and let $T$ be its drawing onto $\pe$ that contains the edge $a_1a_{m+1}$ and is zigzagging (it is easy to see that such a drawing always exists).
By~\cref{lem:unimodal} its cover sequence $\Cover(T)=\{c_i\}_{i=1}^m$ is a unimodal permutation of $\{1,2,\dots,m\}$.
The total length $|T|$ of $T$ can be expressed as $|T|=\sum_{i=1}^m c_i\cdot g_i$.

Now we specify the gap sequence: for $i=1,\dots,m$ set $g_i=c_i$. 
It is easy to see that this sequence defines a plane tree $T$; see \cref{fig:flat-convex-2}. 
It suffices to show that $T$ constitutes the longest plane tree $\Topt$ of $\pe$.

By~\cref{thm:convex}, $\Topt$ is a zigzagging caterpillar.
Also, $a_1a_{m+1}$ is an edge of $\Topt$: suppose not. Since $a_1a_{m+1}$ 
does not cross any other edge, adding it to $\Topt$ produces a plane graph 
with a single cycle $C$. All edges of $\Topt$ are shorter than $a_1a_{m+1}$, 
so omitting any other edge from $C$ yields a longer plane tree, 
a contradiction.

We can thus apply~\cref{lem:unimodal} to see that 
$\Cover(\Topt)$ is a unimodal permutation $\pi$ of $\{1,2,\dots,m\}$ 
and that 
$|\Topt|=\sum_{i=1}^m \pi_i\cdot g_i$. As $c_i$ and $g_i$ match and
as $c$, $g$, and $\pi$ are permuations, 
the Cauchy-Schwarz inequality gives
\begin{equation}
 |\Topt|=\sum_{i=1}^m \pi_i\cdot g_i 
 \le \sqrt{\sum_{i=1}^m \pi_i^2\cdot \sum_{i=1}^m g_i^2} 
 = \sum_{i=1}^m c_i^2 = 
 |T|,
\text{with equality iff $\pi_i=c_i$, for all $i$.}\label{eq:rearrangement}
\end{equation}
Therefore, $\Topt$ is unique and $\Topt=T$ as desired.
\end{proof}

\subsection{Upper bounds on \texorpdfstring{$\bd(d)$}{bd(d)}}

The algorithms for approximating $|\Topt|$ often produce trees with small diameter.
Given an integer $d\ge 2$ and a point set $\pe$, let \( \Td{d}(\pe)\)
be a longest plane tree spanning $\pe$ among those whose diameter is at most $d$.
One can then ask what is the approximation ratio
\[\bd(d) = \inf_{\pe}  \frac{ |\Td{d}(\pe)| }{ |\Topt (\pe)|}
\]
achieved by such a tree. As before, we drop the dependency on $\pe$ in the notation and just use $\Td{d}$ and $\Topt$.

When $d=2$, this reduces to asking about the performance of stars. A result due to Alon, Rajagopalan and Suri~\cite[Theorem 4.1]{AlonRS95} can be restated as $\bd(2)=1/2$.
Below we show a crude upper bound on $\bd(d)$ for general $d$ and then a specific upper bound tailored to the case $d=3$. Note that \cref{thm:dp} shows that $|\Td{3}|$ can be computed in polynomial time.
Our proofs in this section use the notions of flat arc, gap sequence and cover sequence defined in~\cref{sec:convex}.

\thmdiameterbound*
\begin{proof} 
Consider a flat arc on $d+2$ points with gap sequence $G=(1,3,5,\dots,d+1,\dots,6,4,2)$.
Since $G$ is unimodal, we can argue as in the proof of \cref{thm:caterpillars} 
to see that $\Topt$ is the zigzagging caterpillar whose 
cover sequence is $G$, i.e., a path with 
$d+1$ edges (and diameter $d+1$). 
Moreover, this path is the only optimal plane tree spanning the flat arc
because of \cref{thm:convex} and the Cauchy-Schwarz inequality; see the argument leading to (\ref{eq:rearrangement}) 
in the proof of \cref{thm:caterpillars}. 
Therefore, any other plane spanning tree $T\neq\Topt$, be it a zigzagging caterpillar or not, 
must have a length that is an integer less than $|\Topt|$. 
Using that $|\Topt|=\sum_{i=1}^{d+1} i^2= \frac16(d+1)(d+2)(2d+3)= \frac13d^3+o(d^3)$
we obtain
\[ \bd(d)\le  \frac{|\Topt| - 1}{|\Topt|} = 1- \frac{6}{(d+1)(d+2)(2d+3)}. \qedhere
\]
\end{proof}

For $d=3$,~\cref{thm:diameter-bound} gives $\bd(3)\le 29/30$. By tailoring the point set size, the gap sequence $\{g_i\}_{i=1}^m$, and by considering flat convex sets that are not arcs we improve the bound to $\bd(3)\le 5/6$.

\thmdiameterthreebound*
\begin{proof} Consider a flat point set $\pe_{4k+2}$ consisting of two flat arcs that are symmetric with respect to a horizontal line \(\ell'\), each with a gap sequence
\[(\underbrace{1,\dots,1}_{k\times}, \ 2k+1,\ \underbrace{1,\dots,1}_{k\times}).
\]
In other words, $\pe_{4k+2}$ consists of two diametrically opposite points, four unit-spaced arcs of $k$ points each, and a large gap of length $2k+1$ in the middle (see~\cref{fig:diameter-3}).

\begin{figure}[h]
\center
\includegraphics{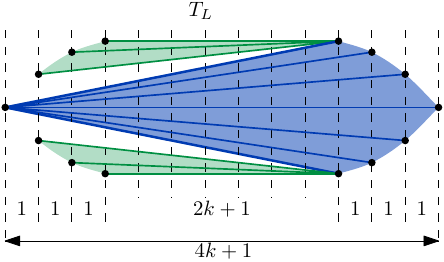}
\caption{An illustration of the point set $\pe_{4k+2}$ when $k=3$ with a tree $T_L$ (blue and green).}
\label{fig:diameter-3}
\end{figure}

On one hand, straightforward counting gives $|\Topt|\ge |T_L|=12k^2+6k+1$, where $T_L$ is the tree depicted in~\cref{fig:diameter-3}.
On the other hand, we claim that any tree $T$ with diameter at most $3$ has length at most $10k^2+6k+1$. Thus
\[ \bd(3)\le \frac{10k^2+6k+1}{12k^2+6k+1},
\]
which tends to $5/6$ as $k\to\infty$.

In the rest of this proof we will show that the longest tree \(T\) among those with diameter at most \(3\) on $\pe_{4k+2}$ has length at most \(10k^2+6k+1\).
First, note that as \(T\) has diameter at most \(3\) it is either a star or it has a cut edge \(ab\) whose removal decomposes \(T\) into a star rooted at \(a\) and a star rooted at \(b\). 

To add the lengths of the edges of the tree, we will often split the edges across the large gap into two parts: one part contained within the left or the right part and with maximum length $k$, and one part going across the large gap, with a length between $2k+1$ and $2k+1+k=3k+1$.

If $T$ is a star, then without loss of generality we can assume that its root $a$ is to the left of the large gap. Denote the distance from the large gap to $a$ by $i$ (we have $0\le i\le k$). Straightforward algebra (see~\cref{fig:flatstar}) then gives

\begin{align*}
|S_a| &= (k-i)^2+i(i+1) \ +\   (2k+1)\cdot (i+2k+1) \ +\ k^2 \\
  &\le k(k+1) \ +\  (2k+1)(3k+1)\  +\  k^2 \\
  &= 8k^2 + 6k +1.
  \end{align*}

\begin{figure}[h]
\center
\includegraphics[page=2]{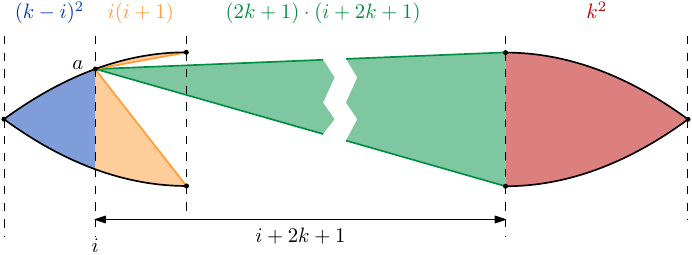}
\caption{The longest star for \(\pe_{4k+2}\) is short.
}
\label{fig:flatstar}
\end{figure}

Now suppose $T$ has diameter $3$ and denote its cut edge by $ab$. If $ab$ is vertical then $|T|=|S_a|$ and the above bound applies.
In the remaining cases, without loss of generality we can assume that $a$ is to the left of $b$. The line $ab$ then splits the remaining points of $\pe_{4k+2}$ into two arcs -- one above $ab$ and the other one below it.

We claim that in the longest tree with diameter $3$, the points of one arc are either all connected to $a$ or all to $b$:
for the sake of contradiction, fix an arc and suppose \(c\) is the rightmost point connected to \(a\). See~\cref{fig:flatunif2}. Then by the choice of \(c\) everything right of it is connected to \(b\) and by convexity everything to the left of \(c\) is connected to \(a\). Let \(\ell\) be the vertical line through the midpoint of \(ab\). If \(c\) lies to the left of \(\ell\), replacing the edge \(ac\) by \(bc\) increases the length of the tree, while keeping the tree plane. Otherwise, connecting the clockwise neighbor of $c$ to $a$ and removing its connection to $b$, increases the length of the tree and keeps it plane. Thus we either have \(c=a\) or \(c=b\) as claimed.

\begin{figure}
\center
\includegraphics{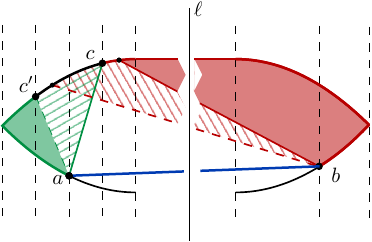}
\caption{Points on one side of $ab$ are either all connected to $a$ or all to $b$.
}
\label{fig:flatunif2}
\end{figure}

Since $\pe_{4k+2}$ is centrally symmetric and we have already dealt with the case when $T$ is a star, we can without loss of generality assume that all points above $ab$ are connected to $a$ and all points below $ab$ are connected to $b$.
Now suppose that $a$ is below the diameter line $\ell'$ of $\pe_{4k+2}$ and consider the reflection $a'\in\pe_{4k+2}$ of $a$ about $\ell'$.
Then the tree $T'$ with cut edge $a'b$ is longer than $T$, since in $T'$ the points to the left of $aa'$ are connected to $b$ rather than to $a$ and all the other edges have the same length (see~\cref{fig:diam3bound:first}). Hence we can assume that $a$ lies above $\ell'$ (or on it). Similarly, $b$ lies below $\ell'$ or on it.

\begin{figure}
\centering
\includegraphics[page=1]{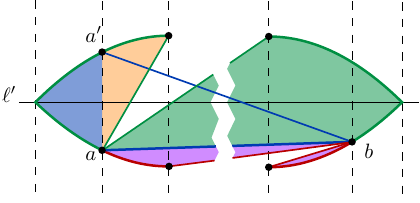}%
\includegraphics[page=2]{figures/flat-abelow.pdf}
\caption{Flipping \(a\) to \(a'\) increases the total length.}
\label{fig:diam3bound:first}
\end{figure}

It remains to distinguish two cases based on whether $a$ and $b$ lie on different sides of the large gap or not.
Either way, denote the distance from the gap to $a$ and $b$ by $i$ and $j$, respectively.
In the first case, considering the subdivision of the edges as sketched in \cref{fig:flatbound}, straightforward algebra gives:

\begin{figure}
\centering
\includegraphics{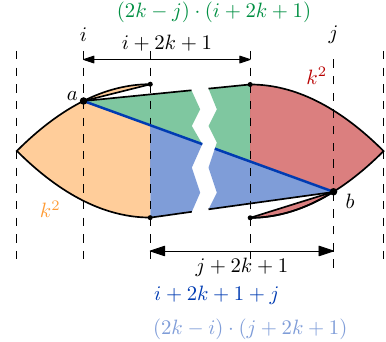}\hfill
\includegraphics[page=2]{figures/flat-situations.pdf}
\caption{Decomposition of the edges for the upper bounds, case 1 is seen left, case 2 is to the right}
\label{fig:flatbound}
\end{figure}

\begin{align*}
|T| &= k^2 \ +\ (2k-j)(i+2k+1) \ +\ (i+2k+1+j) \ +\ (2k-i)(2k+1+j) \ +\ k^2  \\
 &= 10k^2+6k+1 -2ij \le 10k^2+6k+1
\end{align*}
with equality if and only if either $a$ or $b$ (or both) lie on the boundary of the large gap.
In the second case, since $a$ is to the left of $b$ we have $i>j$ and similar algebra, see \cref{fig:flatbound}, gives
\[|T| =  (k^2-2kj+i+2ij) \ +\ (2k+1)(2k+1+i) \ +\  k^2.
\]
Since the right-hand side is increasing in $i$ and $i\le k$, we get
\[|T|\le (k^2+k) + (2k+1)(3k+1) + k^2 = 8k^2 + 6k+1
\]
which is less than the claimed upper bound $10k^2+6k+1$ by a margin.
\end{proof}

\section{Polynomial time algorithms for small diameters}\label{sec:dp}

In \cref{subsec:dpbistar} we show that the longest plane tree of diameter at most $3$ can be computed in
polynomial time using dynamic programming. Our approach bears some resemblance to the polynomial time plane matching algorithm of Aloupis et al.~\cite{aloupis_matching_2010}. The main challenge in our case is the efficient implementation of the dynamic program.
Such a tree may be relevant in providing an approximation algorithm with a better approximation factor.

In \cref{sec:diam4} we show how to compute in polynomial time a longest plane tree among those of the following form: all the points are connected to three distinguished points on the boundary of the convex hull. 
Again, such a tree may play an important role in designing future approximation algorithms, as intuitively it seems to be better than the three stars considered in our previous approximation algorithm, when there is a point in the far region.

\subsection{Finding the longest tree of diameter 3}\label{subsec:dpbistar}
For any two points \(a,b\) of $\pe$, a \emph{bistar rooted} at \(a\) and \(b\) is a tree that contains 
the edge $ab$ and where each point in \(\pe \setminus \{ a,b\} \) is connected to either $a$ or $b$. 
Note that stars are also bistars.
In this section we prove the following theorem.

\thmdp*
\begin{proof}
Note that each bistar has diameter at most 3.
Conversely, each tree with diameter three has one edge \(uv\) where each point \(p\in \pe\) has distance at most \(1\) to either \(u\) or \(v\). It follows that all trees with diameter at most three are bistars.
In \cref{lem:bestbistar} we show that for any two fixed points \(a\), \(b\), the longest plane bistar rooted at  \(a\) and \(b\) can be computed in time \(\bigO(n^2)\).
By iterating over all possible pairs of roots and taking the longest such plane bistar, we find the longest plane spanning tree with diameter at most \(3\) in time \(\bigO(n^4)\).
\end{proof}

We now describe the algorithm used to find the longest plane bistar rooted at two fixed points \(a\) and \(b\).
Without loss of generality, we can assume that the points \(a\) and \(b\) lie on a horizontal line with $a$ to the left of $b$. Furthermore, as we can compute the best plane bistar above and below the line through \(a\) and \(b\) independently, we can assume that all other points lie above this horizontal line. 
In order to solve this problem by dynamic programming, we consider a suitable subproblem.

\begin{figure}
	\centering
	\includegraphics[page=1]{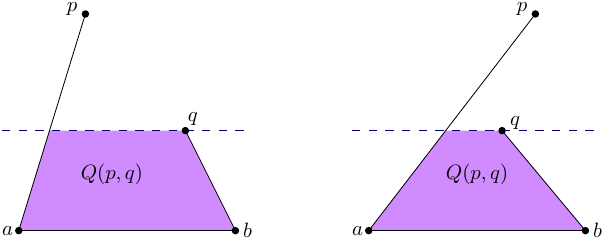}
	\caption{Two examples of valid pairs $p,q$ with their quadrilaterals $Q(p,q)$ shaded.}
	\label{fig:dp_bistar_subproblem}
\end{figure}

The subproblems considered in the dynamic program are indexed by an ordered pair $p,q$ of different points of $\pe$
such that the edges $ap$ and $bq$ do not cross. 
A pair $p,q$ satisfying these condition is a \emph{valid pair}. 
For each valid pair $p,q$, note that the segments $ap$, $pq$, $qb$ and $ba$ form a simple (possibly non-convex) quadrilateral.
Let $Q(p,q)$ be the (convex) portion of this quadrilateral below the horizontal line $y=\min \{y(p),y(q)\}$, as shown in
\cref{fig:dp_bistar_subproblem}.
We define the value \(Z(p,q)\) to be the length of the longest plane bistar rooted at \(a\) and \(b\) on the points in the interior of $Q(p,q)$, without counting $\Vert ab\Vert$.
If there are no points of $\pe$ within the quadrilateral $Q(p,q)$, we set \(Z(p,q) = 0\).

If the quadrilateral $Q(p,q)$ contains some points from $\pe$, we let \(k_{p,q}\) be the highest point of $\pe$ inside of $Q(p,q)$. Then we might connect \(k_{p,q}\) either to \(a\) or to \(b\).
By connecting it to \(a\), we force all points in the triangle \(L_{p,q}\) defined by the edges \(ap\) and \(ak_{p,q}\) 
and the line $y=y(k_{p,q})$, to be connected to \(a\). 
Similarly, when connecting \(k_{p,q}\) to \(b\), we enforce the triangle \(R_{p,q}\) defined by \(bq\), \(bk_{p,q}\) and the line $y=y(k_{p,q})$. See \cref{fig:dp_bistar_recurrence} for an illustration.
In the former case we are left with the subproblem defined by the valid pair $k_{p,q},q$, while in the latter case
we are left with the subproblem defined by the valid pair $p,k_{p,q}$.
Formally, for each valid pair $p,q$ we have the following recurrence:
\begin{align*}
Z(p,q)=\begin{cases} 
0 & \text{if no point of $\pe$ is in $Q(p,q)$,}\\
\max \begin{cases}
Z(k_{p,q},q)+ \Vert ak_{p,q}\Vert + \sum_{l\in L_{p,q}} \Vert al \Vert\\
Z(p,k_{p,q})+ \Vert bk_{p,q}\Vert + \sum_{r\in R_{p,q}} \Vert br \Vert
\end{cases} & \text{otherwise.}
\end{cases}
\end{align*}

\begin{figure}
	\centering 
	\includegraphics[page=2]{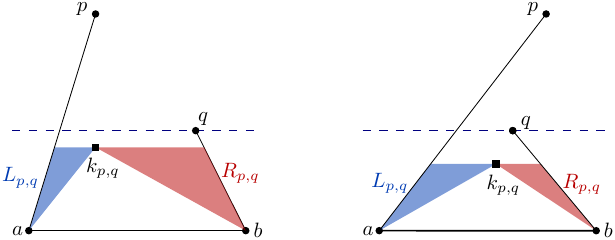}
	\caption{The highest point \(k_{p,k}\) defines two triangular regions and in one of them the edges are forced.}
	\label{fig:dp_bistar_recurrence}
\end{figure}

\begin{lemma}\label{lem:bistarcorrect}
Using \(Z(p,q)\) for all valid $p,q$ we can find a best plane bistar rooted at \(a\) and \(b\).
\end{lemma}
\begin{proof}
Consider a fixed best plane bistar and assume, without loss of generality, 
that the highest point is connected to $a$; the other case is symmetric.
Let $q^*$ be the highest point that is connected to $b$; if $q^*$ does not exist
then the bistar degenerates to a star rooted at $a$. 
This means that all the points above $q^*$ are attached to $a$.
Let $p^*$ be the point above $q^*$ that, circularly around $a$, is closest to $ab$.
See \cref{fig:dp_bistar_p*:a}.
Note that $p^*,q^*$ is a valid pair and all the points above $q^*$ are to the left of $ap^*$.
For $p\in \pe$, denote by $L_p$ the set of points in $\pe$ that, circularly around $a$, are to the left of $ap$.
Similarly, denote by $R_q$ the set of points in $\pe$ below $q$ and to the right of $bq$, when sorted circularly around $b$, see \cref{fig:dp_bistar_p*:b}.
The length of this optimal plane bistar rooted at \(a\) and \(b\) is then
\[
\left(\sum_{l\in L_{p^*}}\Vert al\Vert \right) + \left(\sum_{r\in R_{q^*}}\Vert br\Vert\right) + \Vert ap^*\Vert + \Vert bq^*\Vert + \Vert ab\Vert + Z(p^*,q^*).
\]
On the other hand, each of the values of the form 
\begin{equation}
\left(\sum_{l\in L_p}\Vert al\Vert\right) + \left(\sum_{r\in R_q}\Vert br\Vert\right) + \Vert ap\Vert + \Vert bq\Vert + \Vert ab\Vert + Z(p,q), \label{eq:dp1}
\end{equation}
where $p$, $q$ is a valid pair of points such that $y(p)>y(q)$, is the length of a plane, not necessarily spanning, bistar rooted at \(a\) and \(b\). (It is not spanning, if there is some point above \(q\) and right of  \(ap\).)

Taking the maximum of \(|S_a|\) and equation (\ref{eq:dp1}) over the valid pairs $p,q$ such that 
$y(p)>y(q)$ gives the longest plane bistar for which the highest point is connected to \(a\). A symmetric formula gives the best plane bistar if the highest point is connected to \(b\). Taking the maximum of both cases yields the optimal value.
\begin{figure}
\centering
	\includegraphics[page=3]{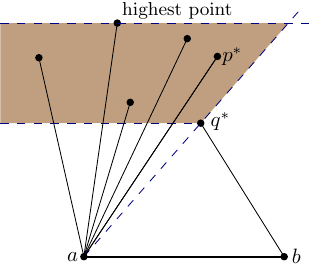}
	\caption{The point $p^*$ is the point of $\pe$ in the shaded region that is angularly around $a$ closest to $q^*$.}
	\label{fig:dp_bistar_p*:a}
	\end{figure}
	\begin{figure}
	\centering
	\includegraphics[page=4]{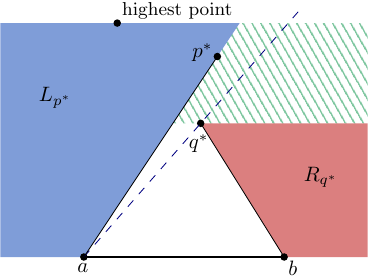}
	\caption{The regions $L_{p^*}$ and $R_{q^*}$ are shaded. The region with bars is empty.}
	\label{fig:dp_bistar_p*:b}
	\end{figure}

The actual edges of the solution can be backtracked by standard methods.
\end{proof}

\begin{lemma}\label{lem:bestbistar}
The algorithm described in the proof of \cref{lem:bistarcorrect} can be implemented in time \(\bigO(n^2)\).
\end{lemma}
\begin{proof}
A main complication in implementing the dynamic program and evaluating (\ref{eq:dp1}) efficiently is finding sums of the form $\sum_{l\in \Delta\cap \pe} \Vert al\Vert$ or
$\sum_{r\in \Delta\cap \pe} \Vert br\Vert$, and the highest point in $\Delta\cap \pe$, 
where $\Delta$ is a query triangle (with one vertex at $a$ or $b$).
These type of range searching queries can be handled using standard 
data structures~\cite{ColeY84,Matousek93}: after preprecessing $\pe$ in time \(\bigO(n^2 \polylog n)\),
any such query can be answered in \(\bigO(\polylog n)\) time.
Noting that there are \(\bigO(n^2)\) such queries, 
a running time of \(\bigO(n^2 \polylog n)\) is immediate. 
However exploiting our specific structure and using careful bookkeeping we can get the running time down to \(\bigO(n^2)\).

As a preprocessing step, we first compute two sorted lists \(\mathcal{L}_a\) and \(\mathcal{L}_b\) of \(\pe \setminus \{a,b\}\). The list \(\mathcal{L}_a\) is sorted angular at \(a\) and \(\mathcal{L}_b\) is sorted by the angle at \(b\).

The values $\sum_{l\in L_p}\Vert al\Vert$ and $\sum_{r\in R_q}\Vert br\Vert$ (as depicted in \cref{fig:dp_bistar_p*:b})
for all $p,q\in \pe$ can be trivially computed in \(\bigO(n^2)\) time:
there are \(\bigO(n)\) such values and for each of them we can scan the whole point set to explicitly
get $L_p$ or $R_q$ in \(\bigO(n)\) time. There are faster ways of doing this, but for our result this suffices.

Assuming that the values $Z(p,q)$ are already available for all valid pairs $p$, $q$, we
can evaluate (\ref{eq:dp1}) in constant time.
For any two points $p$, $q$ we check whether they form a valid pair and whether $y(p)>y(q)$ in constant time.
We can then either evaluate (\ref{eq:dp1}) or the symmetric formula again in constant time.
Thus, in \(\bigO(n^2)\) time we obtain the optimal solution.

It remains to compute the values $Z(p,q)$ for all valid pairs $p$, $q$.
First, we explain how to compute all the triples of the form $(p,q,k_{p,q})$ for all valid pairs $p,q$
in \(\bigO(n^2)\) time.
Then we group these triples in a clever way to implement the dynamic program efficiently.

We focus on the triples with $y(p)< y(q)$ and show how to find the triples of the form $(p,\cdot,\cdot)$ for a fixed \(p\).
The case with \(y(p)> y(q)\) and a fixed \(q\) is symmetric.
Let $W$ be the points of $\pe$ to the right of the ray $bp$ above the horizontal line $y=y(p)$. 
Furthermore let $K$ be the points of $\pe$ to the right of $ap$ and with $y$ coordinate between $y(a)$ and $y(p)$.
An illustration of \(W\) and \(K\) can be found in  \cref{fig:dp_bistar_W}.
Any point $q$ with $y(q)>y(p)$ forms a valid pair $p,q$ if and only if $q$ lies in $W$.
The point $k_{p,q}$ must lie in $K$ by its definition.

\begin{figure}
	\centering 
	\includegraphics[page=5]{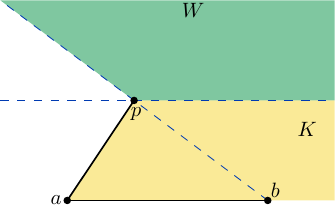}
	\caption{Left: The regions defining $W$ and $K$ for fixed $a$, $b$ and $p$.}
	\label{fig:dp_bistar_W}
\end{figure}

We use \(\mathcal{L}_b\) to find the triples \((p,\cdot, \cdot)\).
We iterate through the list in clockwise order starting at the ray \(ba\) and keep track of the highest \(k^* \in K\) encountered so far.
If the current point lies in \(\pe \setminus (K\cup W)\) we simply skip it.
If the current point lies in \(K\) we update \(k^*\) if necessary.
Finally, if the current point lies in \(W\), we report the triple \((p,q,k^*=k_{p,q})\) with \(q\) being the current point.

For a fixed \(p\) we only iterate \(\mathcal{L}_b\) once.
Thus, for this fixed \(p\) the running time for finding all triples \((p,q,k_{p,q})\) with \(p,q\) forming a valid pair and \(y(p)<y(q)\) is \(\bigO(n)\).
By the procedure we just described and its symmetric procedure on all \(p\in \pe \setminus \{a,b\}\) we find all the triples $(p,q,k_{p,q})$ where $p,q$ is a valid pair in overall \(\bigO(n^2)\) time.

To compute $Z(p,q)$ for all valid pairs using dynamic programing, we also need to compute
the corresponding values $\sum_{l\in L_{p,q}} \Vert al\Vert$ and $\sum_{r\in R_{p,q}} \Vert br\Vert$.
 Refer to \cref{fig:dp_bistar_recurrence} to recall  the definition of \(L_{p,q}\) and \(R_{p,q}\).
For the following procedure we shift the focus to the point \(k\).
For each point $k\in \pe\setminus \{a,b\}$ we collect all triples $(p,q,k=k_{p,q})$,
and compute the sums $\sum_{l\in L_{p,q}} \Vert al\Vert$ and $\sum_{r\in R_{p,q}} \Vert br\Vert$
in linear time for each fixed $k$, as follows. 

We concentrate on the first type of sum, $\sum_{l\in L_{p,q}} \Vert al\Vert$, where $q$
has no role, as the sum is defined by $p$ and $k$.
For the following description we assume that \(\mathcal{L}_a\) is sorted in counterclockwise order.
We create a (sorted) subsequence \(\mathcal{L}_a^k\) of \(\mathcal{L}_a\) containing only the points with \(y\) coordinate below \(y(k)\) and an angle at \(a\) larger than \(\sphericalangle bak\). 
We iterate over the elements of \(\mathcal{L}_a\), starting at the successor of \(k\). While iterating we maintain the last element from \(\mathcal{L}_a^k\) we have seen and the prefix sum \(\sum_{l} \Vert al \Vert\) of all points \(l\) in \(\mathcal{L}_a^k\) we encountered so far.
When advancing to the next point \(p\) in \(\mathcal{L}_a\) there are two possibilities. If \(p\) also lies in \(\mathcal{L}_a^k\) we advance in \(\mathcal{L}_a^k\) and update the prefix sum. Otherwise, we can report \(\sum_{l\in L_{p,q'}}\Vert al \Vert\) to be the current prefix sum for all \(q'\) in a triple \((p,q',k)\).

For any fixed \(k\) this needs at most two iterations through the list \(\mathcal{L}_a\) and can thus be executed in \(\bigO(n)\) time.
A symmetric procedure can be carried out for \(\sum_{r\in R_{p,q}} \Vert br \Vert\) using \(\mathcal{L}_b\).
 As in the case for finding the triples, this results in \(\bigO(n^2)\) overall time to compute all relevant sums $\sum_{l\in L_{p,q}} \Vert al\Vert$ and $\sum_{r\in R_{p,q}} \Vert br\Vert$.

 Now we can implement the dynamic program in the straightforward way. Using the precomputed information we spend \(\bigO(1)\) time for each value \(Z(p,q)\), for a total running time of \(\bigO(n^2)\).
\end{proof}

\subsection{Extending the approach to special trees of diameter 4}
\label{sec:diam4}
Now we show how we can extend the ideas to get a polynomial time algorithm for special trees with diameter \(4\). 
Given three points \(a,b,c\) of $\pe$, a \emph{tristar rooted} at \(a,b\) and \(c\) is a tree
such that each edge has at least one endpoint at $a$, $b$ or $c$. 
We show the following:

\thmtristar*

\begin{proof}
Let \(a,b,c\) be the specified points on the boundary of the convex hull of \(\pe\).
We have to compute the longest plane tristar rooted at \(a,b\) and \(c\).
We assume, without loss of generality, that the edges \(ac\) and \(bc\) are present in the tristar;
the other cases are symmetric. 
We further assume that \(a\) and \(b\) lie on a horizontal line, with $a$ to the left of $b$.

The regions to the left of \(ac\) and to the right of \(bc\), as depicted in \cref{fig:tristarexclude},
can be solved independently of the rest of the instance. Indeed, the presence of the edges \(ac\) and \(bc\)
blocks any edge connecting a point in one of those regions to a point outside the region.
Each one of these regions can be solved as plane bistars, one rooted at $a,c$ and one rooted at $b,c$.
It remains to solve the problem for the points enclosed by $ac$, $cb$ and the portion of the boundary of the convex hull
from $b$ to $a$ (clockwise). 
Let $Q$ be this region. The problem for $Q$ can also be solved independently of the
other problems. We assume for the rest of the argument that all points of $\pe$ are in $Q$.

\begin{figure}
	\centering
	\includegraphics{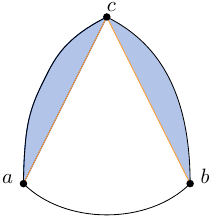}
	\caption{The regions cut off by \(ac\) and \(bc\) can be solved independently as bistars.}
	\label{fig:tristarexclude}
\end{figure}

To solve the problem for $Q$ we will use a variation of the dynamic programming approach for bistars.
For any two points $p,q$ of $\pe$, let us write $p\preceq_c q$ when in the counterclockwise order around 
$c$ we have the horizontal ray through $c$ to the left, then $cp$ and then $cq$
or the segments $cp$ and $cq$ are collinear. 
(Since we assume general position, the latter case occurs when $p=q$.)

The subproblems for the dynamic program are defined by a 5-tuple $(p,p',r,q',q)$ of points of $\pe$
such that 
\begin{itemize}
	\item $p'$ and $q'$ are distinct;
	\item $p\preceq_c p'\preceq_c r \preceq_c q'\preceq_c q$; and
	\item $p'$ and $q'$ are contained in the closed triangle $cpq$.
\end{itemize}
One such five tuple is a \emph{valid tuple},
see \cref{fig:ds-tristar1}.
Note that the first and the second condition imply that $ap$ and $bq$ do not cross.
Some of the points may be equal in the tuple; for example we may have $p=p'$ or $p'=r$ or even $p=p'=r$. 
For each valid tuple $(p,p',r,q',q)$, let $Q(p,r,q)$ be the points of $\pe$ contained in $Q$
below the polygonal path $ap$, $pq$ and $qb$, and below the horizontal line through
the lowest of the points $p,q,r$. Note that the point $r$ is used only to define the horizontal line.

\begin{figure}
	\centering
	\includegraphics[page=1,width=\textwidth]{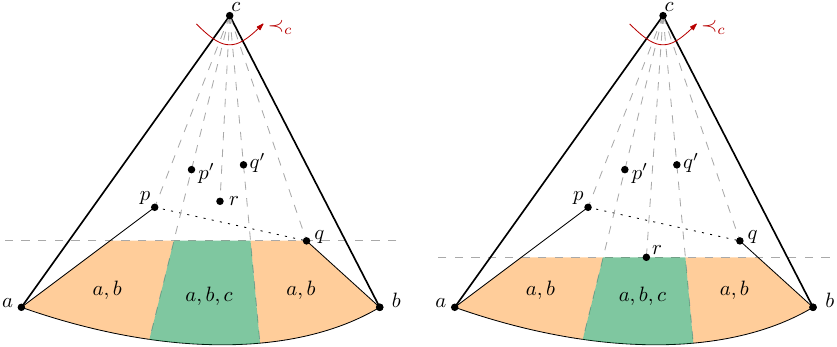}
	\caption{Two examples of valid tuples $(p,p',r,q',q)$ with the region $Q(p,r,q)$ shaded.
		Each $Q(p,r,q)$ is split into 3 regions telling, for each of them, which roots can
		be used for that region.}
	\label{fig:ds-tristar1}
\end{figure}

For each valid tuple $(p,p',r,q',q)$, we define
$Z(p,p',r,q',q)$ as the length of the optimal plane tristar rooted at $a$, $b$, $c$ for the points in the interior of $Q(p,r,q)$, without counting $\Vert ac\Vert+\Vert bc\Vert$, 
and with the additional restriction that a point $k$ can be connected to $c$ only if $p'\preceq_c k\preceq_c q'$.
This last condition is equivalent to telling that no edge incident to \(c\) in the tristar can
cross $ap'$ or $bq'$. 
Thus, we are looking at the length of the longest plane graph
in which each point in the interior of $Q(p,r,q)$ must be connected to either $a$, $b$ or $c$, 
and no edge crosses $ap'$ or $bq'$.

If $Q(p,r,q)$ contains no points, then $Z(p,p',r,q',q)=0$.
If $Q(p,r,q)$ contains some point, we find the highest point \(k= k(p,r,q)\) in \(Q(p,r,q)\) 
and consider how it may attach to the roots and which edges are enforced by each of these choices.
We can always connect \(k \) to \(a\) or \(b\), as \(Q(p,r,q)\) is free of obstacles. 
If $p'\preceq_c k \preceq_c q'$, then we can connect $k$ to $c$. 
If $k$ is connected to $c$ it only lowers the boundary of the region $Q(\cdot,\cdot,\cdot)$
that has to be considered. However if \(k\) is connected to \(a\) or \(b\), it splits off independent regions, some of them can be attached to only one of the roots, some of them are essentially a bistar problem
rooted at $c$ and one of $a,b$. To state the recursive formulas precisely,
for a region $R$ and roots $z,c$, 
let \(\mathrm{BS}_{z,c}(R)\) be the length of the optimal plane bistar rooted at $z,c$ for the points in $R$. 
Such a value can be computed using \cref{lem:bistarcorrect}.

Formally the recurrence for $Z(p,p',r,q',q)$ looks as follows.
If $Q(p,r,q)$ is empty, then $Z(p,p',r,q',q)=0$.
If $Q(p,r,q)$ is not empty, let $k$ be the highest point of $Q(p,r,q)$; 
we have three subcases:
\begin{description}
	\item[Case 1: $k\prec_c p'$.]
		Let $L$ be the points of $Q(p,r,q)$ above $ak$.
		Let \(B\) be the points of \(Q(p,r,q)\) below \(bk\) and
		let \(G\) be the points \(t\) of \(Q(p,r,q)\setminus B\) such that $k\prec_c t \prec_c p'$. Let \(g\) be the point with largest angle at \(b\) among the points in \(G\).
		Then we define \(R'\) to be the set of all points \(t\) in \(Q(p,r,q) \setminus B\) with $p'\prec_c t \prec_c q'$ which are above \(gb\).
		Finally let \(R = Q(p,q,r) \setminus (B\cup R')\),
		as shown in \cref{fig:ds-tristar2}.
		We get the following recurrence in this case:
		\[ 
			Z(p,p',r,q',q) = \max\begin{cases}
				Z(k,p',r,q',q) + \Vert ak \Vert + \sum_{l\in L} \Vert al \Vert\\
				\mathrm{BS}_{a,b}(B)+\Vert bk \Vert + \mathrm{BS}_{b,c}(R') + \sum_{r\in R} \Vert br \Vert
					\end{cases}
		\]
	\item[Case 2: $p' \prec_c k\prec_c q'$.]
		Let $L'$ be the points $t$ of $Q(p,r,q)$ above $ak$ such that $t\prec_c p'$ and
		let $L$ be the points $t$ of $Q(p,r,q)$ above $ak$ such that $p'\prec_c t \prec_c k$.
		Furthermore, let $R$ be the points $t$ of $Q(p,r,q)$ above $bk$ such that $k\prec_c t \prec_c q'$, and 
		let $R'$ be the points $t$ of $Q(p,r,q)$ above $bk$ such that $q' \prec_c t$;
		see \cref{fig:ds-tristar3}.
		We get the following recurrence in this case:
		\[ 
			Z(p,p',r,q',q) = \max\begin{cases}
				Z(p,p',k,q',q) + \Vert ck\Vert \\
				Z(k,k,k,q',q) + \Vert ak \Vert + \mathrm{BS}_{a,c}(L) + \sum_{l\in L'} \Vert al \Vert\\
				Z(p,p',k,k,k) + \Vert bk \Vert + \mathrm{BS}_{b,c}(R) + \sum_{r\in R'} \Vert br \Vert
				\end{cases}
		\]
	\item[Case 3: $q'\prec_c k$.] The case is symmetric to the case $k\prec_c p'$.
\end{description}

\begin{figure}
	\centering
	\includegraphics[page=2,,width=\textwidth]{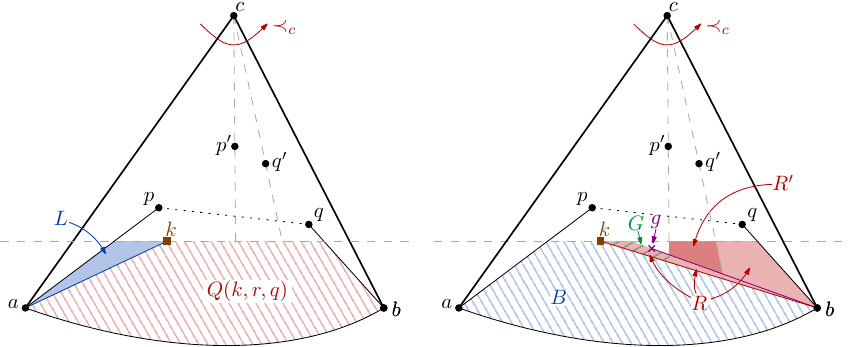}
	\caption{The two cases in the recurrence when $k\prec_c p'$.}
	\label{fig:ds-tristar2}
\end{figure}

\begin{figure}
	\centering
	\includegraphics[page=3]{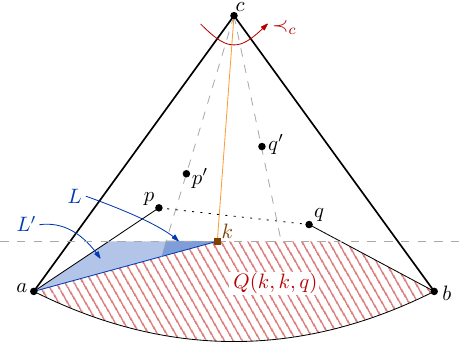}
	\caption{One of the cases in the recurrence when $p' \prec_c k\prec_c q'$.}
	\label{fig:ds-tristar3}
\end{figure}

\begin{figure}
\centering
\includegraphics[page=4]{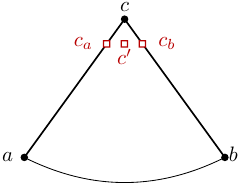}
\caption{Starting case.}
\label{fig:ds-tristar4}
\end{figure}

The values $Z(p,p',r,q',q)$ for all valid tuples $(p,p',r,q',q)$
can be computed using dynamic programming and the formulas described above.
The dynamic program is correct by a simple but tedious inductive argument.

Note that we have to compute a few solutions for bistars. In the recursive calls these have the form $\mathrm{BS}_{z,c}(R)$
 for \(z\in \{a,b\}\) and some region \(R\) of constant size description. More precisely,
each such region is defined by two points of $\pe$ (case of $L$ in \cref{fig:ds-tristar3})
or by four points of $\pe$ (case of $R'$ in \cref{fig:ds-tristar2}, where it is defined by $k,g,q',q$).
Thus, we have $\bigO(n^4)$ different bistar problems, and each of them can be solved in polynomial time
using \cref{lem:bistarcorrect}.

To compute the optimal plane tristar, we add three dummy points to $\pe$
before we start the dynamic programming.
We add a point $c_a$ on the edge $ac$ arbitrarily close to $c$,
a point $c_b$ on the edge $bc$ arbitrarily close to $c$, and a point $c'$
between $c_a$ and $c_b$, see \cref{fig:ds-tristar4}.
We perturb the points $c_a, c'$ and $c_b$ slightly to get them into general position.
Note that \((c_a, c_a, c', c_b, c_b)\) is a valid tuple. We can now compute \(Z(c_a, c_a,c', c_b, c_b)\) by implementing the recurrence with standard dynamic programming techniques in polynomial time. By adding \(\Vert ac \Vert\),  \(\Vert bc \Vert\) and the lengths of the independent bistars from \cref{fig:tristarexclude} to $Z(c_a,c_a,c',c_b,c_b)$, we get the length of the longest plane tristar rooted at \(a,b\) and \(c\).
\end{proof}

Similar to the proof \cref{thm:dp} we can iterate over all possible triples of specified points to find
the best plane tristar rooted at the boundary of the convex hull. 

\section{Local improvements fail}\label{sec:stuck}

One could hope that the longest plane spanning tree problem could 
perhaps be solved by either a greedy approach or by a local search approach~\cite{WilliamsonSh11}.
It is easy to find point sets on as few as 5 points where the obvious greedy algorithm fails to find the longest plane tree.
In this section, we show that the following natural local search algorithm \texttt{AlgLocal($\pe$)} fails too:

\noindent\textbf{Algorithm} \texttt{AlgLocal($\pe$)}:
\begin{enumerate}
\item Construct an arbitrary plane spanning tree $T$ on $\pe$.
\item \textbf{While} there exists a pair of points $a$, $b$ such that $T\cup \{ab\}$ contains an edge $cd$ with $\Vert cd\Vert < \Vert ab\Vert$ and \(T \cup \{ab\} \setminus \{cd\}\) is plane:
  \begin{enumerate}
      \item Set $T\to T\cup \{ab\}\setminus\{cd\}$.\hfill  \texttt{ // tree $T\cup \{ab\}\setminus\{cd\}$ is longer than $T$}
   \end{enumerate}
\item Output $T$.
\end{enumerate}

\begin{lemma}\label{thm:stuck} There are point sets $\pe$ for that the algorithm \texttt{AlgLocal($\pe$)} fails to compute the longest plane tree.
\end{lemma}
\begin{proof}
We construct a point set \(\pe\) consisting of \(9\) points to show the claim.
The first three points of \(\pe\) are the vertices of an equilateral triangle. We assume that two of these vertices are on a horizontal line and denote by \(x\) the center of this triangle. Let \(r_0\) be the circumradius of the triangle.

Inside of this triangle, we want to place the vertices of two smaller equilateral triangles, where again the smaller is contained in the larger one. Set \(\alpha = 17^\circ, r_1 = \frac{2}{3} r_0\) and \(r_2 =\frac{1}{3} r_0\).
 We then place the first of the additional triangles on the circle \(\partial D(x, r_1)\) in a way that the vertices have an angle of \(\alpha\) to the nearest angular bisector of the outer triangle. We place the second additional triangle on the circle \(\partial D(x, r_2)\) again with an angle to the nearest bisector of the outer triangle. However this time, the angle is \(\alpha/2\). This construction is visualized in \cref{fig:stuckconst}.

 \begin{figure}[ht]

\center \includegraphics[page=1]{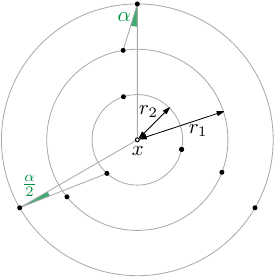}
\caption{Construction of the point set}
\label{fig:stuckconst}

\end{figure}

Now consider the tree on this point set depicted by the solid edges in \cref{fig:stucktree} on the left.
Note that the green, blue and yellow edges are rotational symmetric.
The edges depicted in purple in the tree on the left side of \cref{fig:stucktree} are representatives of the possible edges currently not in the tree, that connect the outer to the inner triangle. However as they are either the smallest edges in the cycle or intersect some edge not in a cycle, the algorithm cannot choose \(ab\) to be one of those edges.

\begin{figure}[ht]
\centering \includegraphics[page=2]{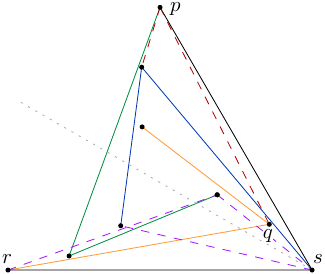}\hfill
\includegraphics[page=3]{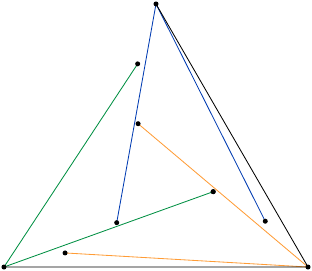}
\caption{Left: A tree which cannot be locally improved; Right: A tree where each pair of edges in the same color is at least as long as the matching pair to the left}
\label{fig:stucktree}
\end{figure}

The possible edges connecting the outer to the middle triangle are depicted in red.
The red edge \(pq \) is not in a cycle with the edge it crosses, and thus it is not a possible swap keeping planarity. 
For the second red edge and the possible edges connecting the middle to the inner triangle, it can similarly be seen that they are shorter than any edge of the current tree. Finally, for the edges along the triangles or the edges connecting the interior and the middle triangle, it can easily be verified that there will be no strictly smaller edge in the unique cycle closed by them.

On the other hand, in the tree depicted in \cref{fig:stucktree} on the right each pair of same colored edges is longer than its counterpart to the left. Therefore \texttt{AlgLocal($\pe$)} does not yield a correct result.
\qedhere

\end{proof}

\section{Conclusions}

We leave several open questions:
\begin{enumerate}
\item What is the correct approximation factor of the algorithm \texttt{AlgSimple($\pe$)} presented in 
	\cref{sec:approximation}? 
While each single lemma in \cref{sec:approximation} is tight for some case, it is hard to believe that the whole analysis, leading to the approximation factor $f\doteq 0.5467$, is tight.  We conjecture that the algorithm has a better approximation guarantee. 
\item What is the approximation factor $\bd(3)$ achieved by the polynomial time algorithm that outputs the longest plane tree with diameter 3? By~\cref{thm:diameter-3-bound} it is at most $5/6$ (and by~\cite{AlonRS95} it is at least $1/2$).
\item For a fixed diameter $d\ge 4$, is there a polynomial-time algorithm that outputs the longest plane tree with diameter at most $d$? By~\cref{thm:dp} we know the answer is yes when $d=3$. And \cref{thm:tristar} gives a positive answer for special classes of trees with diameter \(4\).
Note that a hypothetical polynomial-time approximation scheme (PTAS) has to consider trees of unbounded diameter because of \cref{thm:diameter-bound}. It is compatible with our current knowledge that computing an optimal plane tree of diameter, say, $\bigO(1/\eps)$ would give a PTAS. 
\item Is the general problem of finding the longest plane tree in P? A similar question can be asked for several other plane objects, such as paths, cycles, matchings, perfect matchings, or triangulations. The computational complexity in all cases is open.
\end{enumerate}

\bibliography{MaxEST}

\begin{thebibliography}{10}

\bibitem{AlonRS95}
Noga Alon, Sridhar Rajagopalan, and Subhash Suri.
\newblock Long non-crossing configurations in the plane.
\newblock {\em Fundam. Inform.}, 22(4):385--394, 1995.
\newblock \href {https://doi.org/10.3233/FI-1995-2245}
  {\path{doi:10.3233/FI-1995-2245}}.

\bibitem{aloupis_matching_2010}
Greg Aloupis, Jean Cardinal, S\'ebastien Collette, Erik~D. Demaine, Martin~L.
  Demaine, Muriel Dulieu, Ruy Fabila-Monroy, Vi~Hart, Ferran Hurtado, Stefan
  Langerman, Maria Saumell, Carlos Seara, and Perouz Taslakian.
\newblock Matching points with things.
\newblock In Alejandro L\'opez-Ortiz, editor, {\em {{LATIN}} 2010: Theoretical
  {{Informatics}}}, volume 6034 of {\em Lecture Notes in Computer Science},
  pages 456--467, 2010.
\newblock \href {https://doi.org/10.1007/978-3-642-12200-2_40}
  {\path{doi:10.1007/978-3-642-12200-2_40}}.

\bibitem{Arora98}
Sanjeev Arora.
\newblock Polynomial time approximation schemes for {E}uclidean traveling
  salesman and other geometric problems.
\newblock {\em J. {ACM}}, 45(5):753--782, 1998.
\newblock \href {https://doi.org/10.1145/290179.290180}
  {\path{doi:10.1145/290179.290180}}.

\bibitem{AroraC04}
Sanjeev Arora and Kevin~L. Chang.
\newblock Approximation schemes for degree-restricted {MST} and red-blue
  separation problems.
\newblock {\em Algorithmica}, 40(3):189--210, 2004.
\newblock \href {https://doi.org/10.1007/s00453-004-1103-4}
  {\path{doi:10.1007/s00453-004-1103-4}}.

\bibitem{BarvinokFJTWW03}
Alexander~I. Barvinok, S{\'{a}}ndor~P. Fekete, David~S. Johnson, Arie Tamir,
  Gerhard~J. Woeginger, and Russell Woodroofe.
\newblock The geometric maximum traveling salesman problem.
\newblock {\em J. {ACM}}, 50(5):641--664, 2003.
\newblock \href {https://doi.org/10.1145/876638.876640}
  {\path{doi:10.1145/876638.876640}}.

\bibitem{Biniaz20}
Ahmad Biniaz.
\newblock Euclidean bottleneck bounded-degree spanning tree ratios.
\newblock In {\em Proceedings of the 2020 {ACM-SIAM} Symposium on Discrete
  Algorithms, {SODA} 2020, Salt Lake City, UT, USA, January 5-8, 2020}, pages
  826--836. {SIAM}, 2020.
\newblock \href {https://doi.org/10.1137/1.9781611975994.50}
  {\path{doi:10.1137/1.9781611975994.50}}.

\bibitem{biniaz2020improved}
Ahmad Biniaz.
\newblock Improved approximation ratios for two {E}uclidean maximum spanning
  tree problems.
\newblock \texttt{arXiv:2010.03870}, 2020.

\bibitem{biniaz_maximum_2019}
Ahmad Biniaz, Prosenjit Bose, Kimberly Crosbie, Jean-Lou De~Carufel, David
  Eppstein, Anil Maheshwari, and Michiel Smid.
\newblock Maximum plane trees in multipartite geometric graphs.
\newblock {\em Algorithmica}, 81(4):1512--1534, 2019.
\newblock \href {https://doi.org/10.1007/s00453-018-0482-x}
  {\path{doi:10.1007/s00453-018-0482-x}}.

\bibitem{bloemer_radicals_1991}
Johannes Bl{\"o}mer.
\newblock Computing sums of radicals in polynomial time.
\newblock In {\em 32nd Annual Symposium on Foundations of Computer Science,
  FOCS 1991, San Juan, Puerto Rico, 1-4 October 1991}, pages 670--677. {IEEE}
  Computer Society, 1991.
\newblock \href {https://doi.org/10.1109/SFCS.1991.185434}
  {\path{doi:10.1109/SFCS.1991.185434}}.

\bibitem{cabello2020better}
Sergio Cabello, Aruni Choudhary, Michael Hoffmann, Katharina Klost, Meghana~M
  Reddy, Wolfgang Mulzer, Felix Schr{\"o}der, and Josef Tkadlec.
\newblock A better approximation for longest noncrossing spanning trees.
\newblock In {\em 36th European Workshop on Computational Geometry (EuroCG)},
  2020.

\bibitem{Chan04}
Timothy~M. Chan.
\newblock Euclidean bounded-degree spanning tree ratios.
\newblock {\em Discret. Comput. Geom.}, 32(2):177--194, 2004.
\newblock \href {https://doi.org/10.1007/s00454-004-1117-3}
  {\path{doi:10.1007/s00454-004-1117-3}}.

\bibitem{ChinQW04}
Francis Y.~L. Chin, Jianbo Qian, and Cao~An Wang.
\newblock Progress on maximum weight triangulation.
\newblock In {\em Proc. 10th Annu. Int. Conf. Computing and Combinatorics
  (COCOON)}, volume 3106 of {\em Lecture Notes in Computer Science}, pages
  53--61. Springer, 2004.
\newblock \href {https://doi.org/10.1007/978-3-540-27798-9\_8}
  {\path{doi:10.1007/978-3-540-27798-9\_8}}.

\bibitem{ColeY84}
Richard Cole and Chee{-}Keng Yap.
\newblock Geometric retrieval problems.
\newblock {\em Inf. Control.}, 63(1/2):39--57, 1984.
\newblock \href {https://doi.org/10.1016/S0019-9958(84)80040-6}
  {\path{doi:10.1016/S0019-9958(84)80040-6}}.

\bibitem{CLRS}
Thomas~H. Cormen, Charles~E. Leiserson, Ronald~L. Rivest, and Clifford Stein.
\newblock {\em Introduction to Algorithms}.
\newblock {MIT} Press, 3rd edition, 2009.
\newblock URL: \url{http://mitpress.mit.edu/books/introduction-algorithms}.

\bibitem{deBergChVKOv08}
Mark de~Berg, Otfried Cheong, Marc van Kreveld, and Mark Overmars.
\newblock {\em Computational geometry. Algorithms and applications}.
\newblock Springer-Verlag, Berlin, third edition, 2008.
\newblock \href {https://doi.org/10.1007/978-3-540-77974-2}
  {\path{doi:10.1007/978-3-540-77974-2}}.

\bibitem{DBLP:journals/dcg/DumitrescuT10}
Adrian Dumitrescu and Csaba~D. T{\'{o}}th.
\newblock Long non-crossing configurations in the plane.
\newblock {\em Discrete Comput. Geom.}, 44(4):727--752, 2010.
\newblock \href {https://doi.org/10.1007/s00454-010-9277-9}
  {\path{doi:10.1007/s00454-010-9277-9}}.

\bibitem{EfratIK01}
Alon Efrat, Alon Itai, and Matthew~J. Katz.
\newblock Geometry helps in bottleneck matching and related problems.
\newblock {\em Algorithmica}, 31(1):1--28, 2001.
\newblock \href {https://doi.org/10.1007/s00453-001-0016-8}
  {\path{doi:10.1007/s00453-001-0016-8}}.

\bibitem{Eppstein00}
David Eppstein.
\newblock Spanning trees and spanners.
\newblock In J{\"{o}}rg{-}R{\"{u}}diger Sack and Jorge Urrutia, editors, {\em
  Handbook of Computational Geometry}, pages 425--461. North Holland /
  Elsevier, 2000.
\newblock \href {https://doi.org/10.1016/b978-044482537-7/50010-3}
  {\path{doi:10.1016/b978-044482537-7/50010-3}}.

\bibitem{FranckeH09}
Andrea Francke and Michael Hoffmann.
\newblock The {E}uclidean degree-4 minimum spanning tree problem is {NP}-hard.
\newblock In {\em Proceedings of the 25th {ACM} Symposium on Computational
  Geometry}, pages 179--188. {ACM}, 2009.
\newblock \href {https://doi.org/10.1145/1542362.1542399}
  {\path{doi:10.1145/1542362.1542399}}.

\bibitem{Gilbert79}
Peter~D. Gilbert.
\newblock New results in planar triangulations.
\newblock Technical Report R--850, Univ. Illinois Coordinated Science Lab,
  1979.

\bibitem{Har-Peled11}
Sariel Har-Peled.
\newblock {\em Geometric approximation algorithms}, volume 173 of {\em
  Mathematical Surveys and Monographs}.
\newblock American Mathematical Society, Providence, RI, 2011.
\newblock \href {https://doi.org/10.1090/surv/173}
  {\path{doi:10.1090/surv/173}}.

\bibitem{Klincsek80}
Gheza~Tom Klincsek.
\newblock Minimal triangulations of polygonal domains.
\newblock {\em Ann. Discrete Math.}, 9:121--123, 1980.
\newblock \href {https://doi.org/10.1016/S0167-5060(08)70044-X}
  {\path{doi:10.1016/S0167-5060(08)70044-X}}.

\bibitem{Matousek93}
Jir{\'{\i}} Matou{\v s}ek.
\newblock Range searching with efficient hiearchical cutting.
\newblock {\em Discret. Comput. Geom.}, 10:157--182, 1993.
\newblock \href {https://doi.org/10.1007/BF02573972}
  {\path{doi:10.1007/BF02573972}}.

\bibitem{Mitchell99}
Joseph S.~B. Mitchell.
\newblock Guillotine subdivisions approximate polygonal subdivisions: A simple
  polynomial-time approximation scheme for geometric {TSP}, {$k$-MST}, and
  related problems.
\newblock {\em {SIAM} J. Comput.}, 28(4):1298--1309, 1999.
\newblock \href {https://doi.org/10.1137/S0097539796309764}
  {\path{doi:10.1137/S0097539796309764}}.

\bibitem{Mitchell04}
Joseph S.~B. Mitchell.
\newblock Shortest paths and networks.
\newblock In Jacob~E. Goodman and Joseph O'Rourke, editors, {\em Handbook of
  Discrete and Computational Geometry}, pages 607--641. Chapman and Hall/CRC,
  2nd edition, 2004.
\newblock \href {https://doi.org/10.1201/9781420035315.ch27}
  {\path{doi:10.1201/9781420035315.ch27}}.

\bibitem{MM17}
Joseph S.~B. Mitchell and Wolfgang Mulzer.
\newblock Proximity algorithms.
\newblock In Jacob~E. Goodman, Joseph O'Rourke, and Csaba~D. T\'oth, editors,
  {\em Handbook of Discrete and Computational Geometry}, chapter~32, pages
  849--874. CRC Press, Boca Raton, 3rd edition, 2017.
\newblock \href {https://doi.org/10.1201/9781315119601}
  {\path{doi:10.1201/9781315119601}}.

\bibitem{mulzer04minimum}
Wolfgang Mulzer.
\newblock Minimum dilation triangulations for the regular $n$-gon.
\newblock Master's thesis, Freie Universit{\"a}t Berlin, Germany, 2004.

\bibitem{MulzerOb20}
Wolfgang Mulzer and Johannes Obenaus.
\newblock The tree stabbing number is not monotone.
\newblock In {\em Proceedings of the 36th European Workshop on Computational
  Geometry (EWCG)}, pages 78:1--78:8, 2020.

\bibitem{MulzerR08}
Wolfgang Mulzer and G{\"{u}}nter Rote.
\newblock Minimum-weight triangulation is {NP}-hard.
\newblock {\em J. {ACM}}, 55(2):11:1--11:29, 2008.
\newblock \href {https://doi.org/10.1145/1346330.1346336}
  {\path{doi:10.1145/1346330.1346336}}.

\bibitem{NarasimhanSm07}
Giri Narasimhan and Michiel Smid.
\newblock {\em Geometric spanner networks}.
\newblock Cambridge University Press, Cambridge, 2007.
\newblock \href {https://doi.org/10.1017/CBO9780511546884}
  {\path{doi:10.1017/CBO9780511546884}}.

\bibitem{Papadimitriou77}
Christos~H. Papadimitriou.
\newblock The {E}uclidean traveling salesman problem is {NP}-complete.
\newblock {\em Theor. Comput. Sci.}, 4(3):237--244, 1977.
\newblock \href {https://doi.org/10.1016/0304-3975(77)90012-3}
  {\path{doi:10.1016/0304-3975(77)90012-3}}.

\bibitem{PapadimitriouV84}
Christos~H. Papadimitriou and Umesh~V. Vazirani.
\newblock On two geometric problems related to the traveling salesman problem.
\newblock {\em J. Algorithms}, 5(2):231--246, 1984.
\newblock \href {https://doi.org/10.1016/0196-6774(84)90029-4}
  {\path{doi:10.1016/0196-6774(84)90029-4}}.

\bibitem{RemyS09}
Jan Remy and Angelika Steger.
\newblock A quasi-polynomial time approximation scheme for minimum weight
  triangulation.
\newblock {\em J. {ACM}}, 56(3):15:1--15:47, 2009.
\newblock \href {https://doi.org/10.1145/1516512.1516517}
  {\path{doi:10.1145/1516512.1516517}}.

\bibitem{Welzl92}
Emo Welzl.
\newblock On spanning trees with low crossing numbers.
\newblock In {\em Data structures and efficient algorithms}, volume 594 of {\em
  Lecture Notes in Comput. Sci.}, pages 233--249. Springer, Berlin, 1992.
\newblock \href {https://doi.org/10.1007/3-540-55488-2_30}
  {\path{doi:10.1007/3-540-55488-2_30}}.

\bibitem{WilliamsonSh11}
David~P. Williamson and David~B. Shmoys.
\newblock {\em The design of approximation algorithms}.
\newblock Cambridge University Press, Cambridge, 2011.
\newblock \href {https://doi.org/10.1017/CBO9780511921735}
  {\path{doi:10.1017/CBO9780511921735}}.

\end{thebibliography}

\end{document}